\documentclass[journal]{IEEEtran}

\usepackage{cite}
\usepackage{url}
\usepackage{amsmath}
\usepackage{amssymb}
\usepackage{dsfont}
\usepackage{bm}
\usepackage{enumitem}
\usepackage{mathtools}
\mathtoolsset{showonlyrefs} 
\usepackage[caption=false,font=footnotesize]{subfig}
\usepackage[most]{tcolorbox}
\tcbset{colback=white, colframe=black, boxrule=0.5pt,}

\DeclareMathOperator*{\minimize}{minimize}
\DeclareMathOperator*{\st}{subject\ to}

\def\ra{\rightarrow}

\def\argmin{\text{argmin}}
\def\argmax{\text{argmax}}
\def\diag{\text{diag}}

\def\ln{\text{ln}}

\newcommand{\mathletter}[1]{%
	\expandafter\newcommand\csname b#1\endcsname{\mathbb #1}
	\expandafter\newcommand\csname c#1\endcsname{\mathcal #1}
	\expandafter\newcommand\csname f#1\endcsname{\mathfrak #1}
	\expandafter\newcommand\csname til#1\endcsname{\widetilde #1}
	\expandafter\newcommand\csname ha#1\endcsname{\widehat #1}
	\expandafter\newcommand\csname bf#1\endcsname{\bf #1}
}%
\def\mathletters#1{\mathlettersB #1,,}
\def\mathlettersB#1,{\ifx,#1,\else\mathletter #1\expandafter\mathlettersB\fi}
\mathletters{A,B,C,D,E,F,G,H,I,J,K,L,M,N,O,P,Q,R,S,T,U,V,W,X,Y,Z}

\def \qed {\hfill \vrule height6pt width 6pt depth 0pt}
\def\bee{\begin{equation}}
\def\ene{\end{equation}}
\def\beq{\begin{eqnarray}}
\def\enq{\end{eqnarray}}

\newtheorem{defi}{Definition}
\newtheorem{theo}{Theorem}
\newtheorem{lemma}{Lemma}
\newtheorem{assum}{Assumption}
\newtheorem{coro}{Corollary}
\newtheorem{remark}{Remark}

\newtheorem{example}{Example}
\newenvironment{proof}{\begin{IEEEproof}}{\end{IEEEproof}}

\def\bone{{\mathds{1}}}

\title{\huge\bf Distributed Discrete-time Optimization in Multi-agent Networks Using only Sign of Relative State}

\author{Jiaqi~Zhang, Keyou~You, ~\IEEEmembership{Senior Member,~IEEE}, and~Tamer~Ba\c{s}ar, ~\IEEEmembership{Life Fellow,~IEEE}
\thanks{*This work was in part supported by the National Natural Science Foundation of China (61722308), Tsinghua University Initiative Scientific Research Program, and in part by ARL under cooperative agreement W911NF-17-2-0196 ({\em Corresponding author: Keyou You}).}
\thanks{Jiaqi Zhang and Keyou You are with the Department of Automation, and BNRist, Tsinghua University, Beijing 100084, China. E-mail: zjq16@mails.tsinghua.edu.cn, youky@tsinghua.edu.cn.} %
\thanks{Tamer Ba\c{s}ar is with the Coordinated Science Laboratory, University of Illinois at Urbana-Champaign, Urbana, IL 61801 USA. E-mail: basar1@illinois.edu.}
}

\begin{document}

\maketitle

\begin{abstract}
This paper proposes distributed discrete-time algorithms to cooperatively solve an additive cost optimization problem in multi-agent networks. The striking feature lies in the use of only the sign of relative state information between neighbors, which substantially differentiates our algorithms from others in the existing literature. We first interpret the proposed algorithms in terms of the penalty method in optimization theory and then perform non-asymptotic analysis to study convergence for static network graphs. Compared with the celebrated distributed subgradient algorithms, which however use the exact relative state information, the convergence speed is essentially not affected by the loss of information. We also study how introducing noise into the relative state information and randomly activated graphs affect the performance of our algorithms.  Finally, we validate the theoretical results on a class of distributed quantile regression problems. 
\end{abstract}

\begin{IEEEkeywords}
Distributed optimization, multi-agent networks, sign of relative state, penalty method, subgradient iterations.
\end{IEEEkeywords}

\section{INTRODUCTION}\label{sec_1}

\IEEEPARstart{R}{ecently}, there has been an increasing interest in distributed optimization problems in multi-agent networks. Distributed optimization requires all agents to cooperatively minimize a sum of local objective functions under the constraint that each agent only obtains its local objective function.  Thus agents must exchange information with their neighbors to find an optimal solution. The motivating examples include formation control \cite{you2011network,cao2013overview}, large scale machine learning \cite{cevher2014convex,you2018distributed}, and distributed quantile regression over sensor networks \cite{wang2017distributed}. An overview of this topic can be found in \cite{nedic2017network}.

Many existing algorithms to solve distributed optimization in multi-agent networks generally comprise two parts, see e.g. \cite{nedic2009distributed,nedic2017network,nedic2015distributed,shi2015extra,yi2014quantized,pu2017quantization,rabbat2005quantized,zhu2016quantized} and the references therein. One is to drive all agents to reach a consensus, and the other is to push the consensus value toward an optimal solution  of the optimization problem.  However,  they all require each agent to access the {\em exact} relative state information with respect to its neighbors \cite{nedic2009distributed,nedic2017network,nedic2015distributed,shi2015extra} or the quantized {\em absolute} state \cite{yi2014quantized,pu2017quantization,rabbat2005quantized}. In some applications, however, an agent is only able to acquire a very rough {\em relative} state information with respect to its neighbors.  As a notable example, consider several working robots in a horizontal line, where each robot can only decide whether a neighbor is on its left side or right side. In this case, each agent can only access one bit of relative state information from each of its neighbors. Clearly,  this is very different from the quantized settings in \cite{yi2014quantized,pu2017quantization,rabbat2005quantized,zhu2016quantized}, which use the quantized version of the absolute state, and dynamic quantizers are essential for computing an exact optimal solution \cite{yi2014quantized}. With static quantizers, each node can only find a sub-optimal solution \cite{pu2017quantization,rabbat2005quantized,zhu2016quantized}. This also distinguishes our work from \cite{seide20141,magnusson2017convergence}  where one-bit quantized gradients are used. We show in this work that knowing only the sign of relative state (which is essentially\footnote{We say ``essentially'' because sign function takes also the 0 value, in addition to 1 and -1. But, in any implementation, the ``0'' value appears very rarely.} one bit information for each neighbor) is sufficient to obtain an exact optimal solution.  Other distributed optimization algorithms include the ADMM-based methods \cite{Ling2014Decentralized,shi2014linear,zhu2016quantized}, and proximal gradient methods \cite{aybat2017distributed}. Note that these algorithms need much more than one bit of information per time from its neighbors.

To the best of our knowledge, the use of one bit relative state information in distributed algorithms has been previously studied in \cite{chen2011finite,chen2012distributed,guo2013consensus,franceschelli2015finite,franceschelli2017finite,lin2017distributed}, but in the context of different problems. Particularly, the authors in \cite{chen2011finite,guo2013consensus,franceschelli2015finite,franceschelli2017finite} are concerned with the consensus problem by using sign of the relative state. Except \cite{lin2017distributed}, the underlying problem is not an optimization problem.  Moreover, all these works study distributed algorithms in the continuous-time regime, and adopt the well-established non-smooth analysis tools  \cite{clarke2008nonsmooth} to analyze convergence. 

Discrete-time algorithms are worth studying for distributed optimization in multi-agent networks. First, many  applications of distributed optimization involve communication between agents and control of agents, which are typically discrete in nature. Second, discrete-time algorithms are easier to implement than their continuous-time versions \cite{lin2017distributed}. Third, neither the non-smooth analysis tools nor the Lyapunov-based methods for continuous-time algorithms \cite{lin2017distributed} are applicable to the discrete-time case. Specifically, the rule of thumb for selecting stepsize in discrete-time algorithms cannot guarantee the existence of a valid Lyapunov function, and the sophisticated stepsize rules (e.g. line minimization rule) cannot be easily implemented in a distributed manner. Thus the Lyapunov-based methods \cite{lin2017distributed} seem impossible to extend to the discrete-time case. Finally, the continuous-time multi-agent networks with one bit of feedback information renders the common numerical methods, e.g. the Euler discretization,  inapplicable \cite{dieci2012survey}. That is, simple discretization of the continuous-time algorithm may lead to an ill-posed discrete-time algorithm. Accordingly, an alternative method of approach and analysis is needed, which is the primary objective of this work.

This paper proposes distributed discrete-time optimization algorithms in multi-agent networks that use for each agent only the sign of relative state value for each neighbor. We first interpret the distributed algorithms by the penalty method in optimization theory \cite{bertsekas2015convex}, and show that they are the exact subgradient iterations of a penalized optimization problem, which is specially designed in conformity with the network structure. An interesting finding is that the finite penalty factor can be explicitly given in terms of the network size and its connectivity. This allows us to analyze the convergence of the discrete-time algorithms in a substantially different way as compared with previous works \cite{lin2017distributed,franceschelli2017finite}.  In particular, our analysis is based on optimization theory rather than algebraic graph theory or Lyapunov theory. The advantages of such an approach are at least twofold. First, compared to many existing approaches which first propose an algorithm and then find a Lyapunov function to prove its convergence, the intuition behind our algorithm appears more natural and reasonable, as it aims to minimizing a well-designed objective function. Second, a wealth of research in optimization theory is directly applicable to our algorithms, making it natural and quite easier to handle other scenarios, e.g., random network graphs and the sign of perturbed relative state, both of which are investigated in this work.  

We also provide non-asymptotic results to describe the behavior of our distributed algorithms under diminishing stepsizes as well as a constant stepsize. This implies that the convergence rate of the objective function for diminishing stepsizes varies from $O(1/\ln(k))$ to $O(\ln(k)/\sqrt{k})$, depending on the choice of the stepsize, where $k$ is the number of iterations. It should be noted that  $O(\ln(k)/\sqrt{k})$ is an optimal rate for a generic subgradient algorithm; see, for example, Page 9 of \cite{boyd2011subgradient}. That is, our distributed algorithms with only sign information on the relative state essentially do not lead to any reduction in the convergence rate. Different from \cite{nedic2009distributed},  the convergence under diminishing stepsizes does not require uniform boundedness of the subgradient of the objective function. For a constant stepsize, it approaches a neighborhood of an optimal solution at a rate $O(1/k)$ and the error is proportional to the stepsize.

Notably, in real applications, the relative state information is often obtained via communication networks or sensors, and is typically noise corrupted. This results in each node unable to obtain the sign of the relative state accurately. A natural question that comes up is how the noise of this type affects  the performance of distributed optimization algorithms. In the context of consensus seeking,  this problem has been extensively studied, see e.g., \cite{liu2011distributed,kar2009distributed,cheng2016convergence}.  Since consensus algorithms are linear and do not involve optimization, the approaches in these papers do not apply to the current setting. Here we also adopt an optimization based approach to study the performance of our distributed algorithm when the relative state is corrupted by Gaussian noise, showing the robustness of the algorithm. 

Subsequently, we extend the above results to randomly activated network graphs, which are known as gossip-like graphs \cite{dimakis2010gossip,kan2016leader}, and show that the distributed algorithms over random graphs are the exact stochastic subgradient iterations of a penalized optimization problem.  Note that the results for continuous-time counterpart in \cite{lin2017distributed,franceschelli2017finite} are limited to static network graphs, and it is unclear whether they can be extended to time-varying graphs via the approaches employed there. 

Finally, we apply our algorithms to solve a distributed quantile regression problem. Clearly, this problem is of independent interest and it has already been studied in \cite{wang2017distributed} using the distributed subgradient algorithm of \cite{nedic2009distributed}.  We approach that problem using our framework and theory, and confirm that the distributed quantile regression can be well solved using only sign of relative state. Compared with \cite{wang2017distributed}, the feedback information from each neighbor is now reduced to essentially only one bit at every node. 

Some results in this paper are obtained in \cite{zhang2018distributed}, where it requires the uniform boundedness of the subgradient of the objective function and omits the proof of its major result. This paper further considers the cases under a constant stepsize and the noisy measurement. 

The rest of the paper is organized as follows. Section \ref{sec_2} formulates the distributed optimization problem. In Section \ref{sec_3}, we present our discrete-time distributed optimization algorithm that uses only the sign of neighbor relative state and interpret it as subgradient iterations of a penalized optimization problem. Section \ref{sec_4} performs non-asymptotic analysis on the distributed algorithm under diminishing stepsizes as well as a constant stepsize.  In Section \ref{sec_7}, we examine the performance of our algorithm with relative measurement errors.  We then propose a modified algorithm to solve the problem over randomly activated graphs in Section \ref{sec_5}. Section \ref{sec_6} introduces the distributed quantile regression problem, which is solved using our algorithms, which also validates our theoretical results. Some concluding remarks are drawn in Section \ref{sec_con}. The paper ends with two appendices, which contain proofs of two of the main theorems.

\textbf{Notation}: We use $a,\bm a,A$ and $\cA$ to denote a scalar, vector, matrix and set, respectively. $\bm a^\mathsf{T}$ and $A^\mathsf{T}$ denote the transposes of $\bm a$ and $A$, respectively. $\bR$ denotes the set of real numbers and $\bR^n$ denotes the set of all $n$-dimensional real vectors. $\bone$ denotes the vector with all ones, the dimension of which depends on the context. Let $\|\cdot\|_1,\|\cdot\|$ and $\|\cdot\|_\infty$ denote the $l_1$-norm, $l_2$-norm and $l_{\infty}$-norm of a vector or matrix, respectively. We define
\[
\text{sgn}(x)=\left\{
	\begin{array}{cl}
	1,&\text{if }x>0,\\ \relax
	0,&\text{if }x=0,\\
	-1,& \text{if }x<0.
	\end{array}\right.
\]

With a slight abuse of notation, $\nabla f(x)$ denotes \emph{any} subgradient of $f(x)$ at $x$, i.e., $\nabla f(x)$ satisfies 
\bee\label{subgradient}
f(y)\geq f(x)+(y-x)^\mathsf{T}\nabla f(x),\ \forall y\in \bR.
\ene
The subdifferential $\partial f(x)$ is the set of all subgradients of $f(x)$ at $x$. If $f(x)$ is differentiable at $x$, then $\partial f(x)$ includes only the gradient of $f(x)$ at $x$. Superscripts are used to represent sequence indices, i.e., $x^k$ represents the value of the sequence $x$ at time $k$.

\section{Problem Formulation}\label{sec_2}
This section introduces some basics of graph theory, and presents the distributed optimization problem in multi-agent networks.
\subsection{Basics of Graph Theory}
A graph (network) is represented as $\cG=(\cV,\cE)$, where $\cV=\{1,...,n\}$ is the set of nodes and $\cE
\subseteq \cV \times \cV$ is the set of edges. Let $\cN_i=\{j\in\cV|(i,j)\in \cE\}$ be the set of neighbors of node $i$, and $A=[a_{ij}]$ be the weighted adjacency matrix of $\cG$, where $a_{ij}>0$ if there exists an edge connecting nodes $i$ and $j$, and otherwise, $a_{ij}=0$. If $A=A^\mathsf{T}$, the associated graph is undirected. This paper focuses only on undirected graphs. A path is a sequence of consecutive edges. We say a graph is {\em connected} if there exists a path between any pair of nodes. We introduce an important concept called $l$-connected graph.
\begin{defi}[$l$-connected graph]
A connected graph is $l$-connected ($l\geq 1$) if it remains connected whenever fewer than $l$ edges are removed. 
\end{defi}

Clearly each node of an $l$-connected graph has at least $l$ neighbors.

\subsection{Distributed Optimization Problem}
With only the sign of relative state, our objective is to distributedly solve the multi-agent optimization problem
\bee\label{original}
\minimize_{x\in\bR}\ f(x):=\sum_{i=1}^n f_i(x)
\ene
where for each $i\in\cV$, the local objective function $f_i(x)$ is continuously convex but not necessarily differentiable, and is only known by node $i$. The number of nodes is set to be $n>1$.  We first make a standard assumption.
\begin{assum}\label{assum1}  The set $\cX^\star$ of optimal solutions of problem \eqref{original} is nonempty, i.e., for any $x^\star\in \cX^\star$, it holds that $f^\star:=f(x^\star)=\inf_{x\in\bR}f(x)$.
\end{assum}

\section{A Distributed Optimization Algorithm over Static Graphs}\label{sec_3}
In this section, we propose our discrete-time distributed optimization algorithm that  uses only sign information of the relative state of the neighboring nodes (which we call, by a slight abuse of terminology, ``one bit information''), and then interpret it via the penalty method in optimization theory. 
\subsection{The Distributed Optimization Algorithm}
Our distributed algorithm to solve \eqref{original} over a static network $\cG$ is given as follows. For all $i\in\cV$,
\begin{tcolorbox}[ams equation,size=small]\label{protocol}\tag{Algo. 1}
x_i^{k+1}=x_i^{k}+\lambda\rho^k\sum_{j\in \cN_i}a_{ij}\text{sgn}(x_j^k-x_i^k)-\rho^k \nabla f_i(x_i^k)
\end{tcolorbox}
\noeqref{protocol}
\noindent where $x_i^k$ is the state of node $i$, $\lambda$ is a positive scalar, $\rho^k$ is the stepsize, $\cN_i$ is the set of neighbors of node $i$, and $\nabla f_i(x_i^k)$ is any subgradient of $f_i(x)$ at $x_i^k$, see \eqref{subgradient}.  

The continuous-time version of \ref{protocol} is given in \cite{lin2017distributed}. To ensure a valid algorithm, it is important to choose both $\lambda$ and $\rho^k$, which, for the discrete-time case, requires a completely different approach from that of \cite{lin2017distributed}, as it will be evident in Section \ref{sec_pmi}.

Compared with the celebrated distributed subgradient descent algorithm, see e.g.\cite{nedic2009distributed},  
\bee\label{dgd}
\begin{aligned}
x_i^{k+1}=x_i^{k}+\sum_{j\in \cN_i}a_{ij}(x_j^k-x_i^k)-\rho^k \nabla f_i(x_i^k)
\end{aligned}
\ene
\ref{protocol} only uses $\text{sgn}(x_j^k-x_i^k)$ instead of the exact relative state $(x_j^k-x_i^k)$. Thus, each node needs only to know the sign of the relative state, which is clearly the minimum information and can be easily extended to the case of multi-level quantization (that is, multiple bits).     
\begin{remark}
\ref{protocol} also works if $x$ is a vector by applying $\text{sgn}(\cdot)$ to each element of the relative state vector. All the results on the scalar case continue to hold with such an adjustment.
\end{remark}

%

\subsection{Penalty Method Interpretation of \ref{protocol}}
\label{sec_pmi}
In this subsection, we interpret \ref{protocol} via the penalty method and show that it is the subgradient iteration of a penalized optimization problem. 

Notice that problem \eqref{original} can be essentially reformulated as follows:
\bee\label{reformulate}
\begin{aligned}
	&\minimize_{\bm x\in\bR^n}&&g(\bm x):=\sum_{i=1}^n f_i(x_i)\\
	&\st &&x_i=x_j,\ \forall i,j\in\{1,...,n\}
\end{aligned}
\ene
where $\bm x=[x_1,...,x_n]^\mathsf{T}$. It is easy to see that the optimal value of problem \eqref{reformulate} is also $f^\star$, and the set of optimal solutions is $\{x^\star\bone|x^\star\in\cX^\star\}$. Define a penalty function by
\bee\label{hx1}
h(\bm x)=\frac{1}{2}\sum_{i=1}^n\sum_{j\in \cN_i}a_{ij}|x_i-x_j|. 
\ene

If the associated network $\cG$ is connected, then $h(\bm x)=0$ is equivalent to that $x_i=x_j,\ \forall i,j\in\{1,...,n\}$. Thus, a penalized optimization problem of \eqref{reformulate} can be given as
\bee\label{penalty}
\minimize_{\bm x\in\bR^n}\ \tilde{f}_\lambda(\bm x):=g(\bm x)+\lambda h(\bm x)
\ene
where $\lambda>0$ is the penalty factor.

We show below that \ref{protocol} is just the subgradient iteration of the penalized problem \eqref{penalty} with stepsizes $\rho^k$. Recall that $\text{sgn}(x)$ is a subgradient of $|x|$ for any $x\in\bR$. It follows from \eqref{hx1} that a subgradient $\nabla h(\bm x)=[\nabla h(\bm x)_1,...,\nabla h(\bm x)_n]^\mathsf{T}$ of $h(\bm x)$ is given element-wise by
\bee\nonumber
\begin{aligned}
\nabla h(\bm x)_i &=\sum_{j\in \cN_i}a_{ij}\text{sgn}(x_i-x_j),\ i\in\cV.
\end{aligned}
\ene
Similarly, a subgradient $\nabla g(\bm x)=[\nabla g(\bm x)_1,...,\nabla g(\bm x)_n]^\mathsf{T}$ of $g(\bm x)$ is given element-wise by
$
\nabla g(\bm x)_i=\nabla f_i(x_i).
$
Then, the $i$-th element of a subgradient of $\tilde{f}_\lambda(\bm x)$ is given as
\bee\label{grad_f2}
\nabla \tilde{f}_\lambda(\bm x)_i=\lambda\sum_{j\in \cN_i}a_{ij}\text{sgn}(x_i-x_j)+\nabla f_i(x_i), i\in\cV.
\ene
Finally, the subgradient method for solving \eqref{penalty} is given as
\bee\label{subgrad}
\bm x^{k+1}=\bm x^k-\rho^k\nabla \tilde{f}_\lambda(\bm x^k),
\ene
which is exactly the vector form of \ref{protocol}. By \cite{bertsekas2015convex}, it follows that the subgradient method converges to an optimal solution of problem \eqref{penalty} if $\rho^k$ is appropriately chosen. 

For a finite $\lambda>0$, the optimization problems \eqref{reformulate} and \eqref{penalty} are generally not equivalent. Under mild conditions, however, we prove that they actually become equivalent if the penalty factor $\lambda$ is strictly greater than an explicit lower bound. 

\begin{assum}\label{assum2}
\begin{enumerate}[label=(\alph*)]
\item\label{assum2a} (Uniform Boundedness) There exists a $c>0$ such that
\bee\label{uppergrad}
|\nabla f_i(x)|\leq c,\ \forall i\in\cV,x\in\bR.
\ene
\item\label{assum2b} There exist $ c>0$ and $\alpha>0$ such that
\bee\nonumber
|\nabla f_i(x)|^2\leq\frac{1}{2} c^2(\alpha+\min_{x^\star\in\cX^\star}|x-x^\star|^2),
\ \forall i\in\cV,x\in\bR.
\ene
\end{enumerate}
\end{assum}

Assumption \ref{assum2}\ref{assum2a} is often made to guarantee the convergence of a subgradient method \cite{nedic2009distributed}, and holds if $\{\bm x^k\}$ is restricted to a compact set. Assumption \ref{assum2}\ref{assum2b} is obviously weaker than Assumption \ref{assum2}\ref{assum2a}, and holds if $f_i(x)$ is quadratic. Then, it is easy to obtain the following two results, proofs of which are quite straightforward and are therefore not included.
\begin{enumerate}[label=(\alph*)]
\item Under Assumption \ref{assum2}\ref{assum2a}, we have that
\bee\label{c_a}
\|\nabla \tilde{f}_\lambda(\bm x)\|\leq c_a, \forall \bm x\in\bR^n
\ene
where $c_a= \sqrt{n}( c+\lambda \|A\|_\infty)$.
\item Under Assumption \ref{assum2}\ref{assum2b}, we have that
\bee\label{c_b}
\begin{aligned}
\|\nabla \tilde{f}_\lambda(\bm x)\|^2&\leq c_b^2+ c^2\min_{x^\star\in\cX^\star}\|\bm x-x^\star\bone\|^2, \forall \bm x\in\bR^n
\end{aligned}
\ene
where $c_b= \sqrt{n\alpha c^2+2\lambda^2 \|A\|_\infty^2}$.
\end{enumerate}

Now we are ready to present the main result of this subsection. To this end, we define
\bee\label{eq1_theo1}
\begin{aligned}
\bar{x}&=\frac{1}{n}\bone^\mathsf{T}\bm x,\\
v(\bm x)&=\max_i (x_i)-\min_i (x_i),
\end{aligned}
\ene
and let $a_\text{min}^{(l)}$ be the sum of the $l$ smallest edges' weights, i.e. 
\bee
a_\text{min}^{(l)}=\sum_{e=1}^l a_{(e)}\label{minsum}
\ene where $a_{(1)},a_{(2)}, \ldots $ are an ascending order of the positive weights $a_{ij}, \forall (i,j)\in\cE$. 
\begin{theo}\label{theo1}
Suppose that Assumptions \ref{assum1} and \ref{assum2}\ref{assum2a} hold, and that the multi-agent network is $l$-connected.  If the penalty factor is selected as
\bee\label{lowerbound}
\lambda>\underline{\lambda}:=\frac{n c}{2a_\text{min}^{(l)}},
\ene
where $ c$  and $a_\text{min}^{(l)}$ are defined in    \eqref{uppergrad} and \eqref{minsum}, then:
\begin{enumerate}[label=(\alph*)]
\item The optimization problems \eqref{original} and \eqref{penalty} are equivalent in the sense that the set of optimal solutions and optimal value of  \eqref{penalty} are given by $ \tilde{\cX}^\star=\{x^\star\bone|x^\star\in \cX^\star\}$ and $f^\star$ respectively.

\item For any $\bm x \notin\{\alpha\bone| \alpha\in\bR\}$, it holds that
\[\|\nabla \tilde{f}_\lambda(\bm x)\|_\infty\geq \frac{2\lambda a^{(l)}_\text{min}}{n}- c.\]
\end{enumerate}
\end{theo}
\begin{proof}  (of part (a))~ Consider the  inequality below
\bee\label{eq1_theo3}
\begin{aligned}
\tilde{f}_\lambda(\bm x)&=\lambda h(\bm x)+g(\bm x-\bar{x}\bone+\bar{x}\bone)\\
&\geq \lambda h(\bm x)+g(\bar{x}\bone)+(\bm x-\bar{x}\bone)^\mathsf{T}\nabla g(\bar{x}\bone)\\
&\geq \lambda h(\bm x)+f(\bar{x})-\|\bm x-\bar{x}\bone\|\cdot\|\nabla g(\bar{x}\bone)\|
\end{aligned}
\ene
where the equality follows from the definition of $\tilde{f}_\lambda(\bm x)$, the first inequality is from \eqref{subgradient}, and the second inequality results from the Cauchy-Schwarz inequality as well as the fact that $g(a\bone)=f(a)$.

Then, we can show that
\beq\label{eq2_theo3}
h(\bm x)\geq a_\text{min}^{(l)}v(\bm x).
\enq
Since the multi-agent network is $l$-connected, it follows from Menger's theorem \cite{deo1974graph} that there exist at least $l$ disjoint paths (two paths are disjoint if they have no common edge) between any two nodes of the graph. Therefore, letting $x_\text{max}$ and $x_\text{min}$ be two nodes associated with the maximum element and the minimum element of $\bm x$, respectively, we can find $l$ disjoint paths from $x_\text{max}$ to $x_\text{min}$. Let $x_{(p,1)},...,x_{(p,n_p)}$ denote the nodes of path $p$ in order, where $n_p$ is the number of nodes in path $p$, and $x_{(p,1)}=x_\text{max}$, $x_{(p,n_p)}=x_\text{min}$ for all $p\in\{1,...,l\}$. Since these $l$ paths are disjoint, it follows that
\bee\label{eq5_theo1}
\begin{aligned}
h(\bm x)&\geq\sum_{p=1}^{l}\sum_{i=1}^{n_p-1}a_{(p,i,i+1)}|x_{(p,i)}-x_{(p,i+1)}|\\
&\geq\sum_{p=1}^{l}\sum_{i=1}^{n_p-1}\min_{i}a_{(p,i,i+1)}|x_{(p,i)}-x_{(p,i+1)}|\\
&\geq\sum_{p=1}^{l}\min_{i}a_{(p,i,i+1)}\sum_{i=1}^{n_p-1}(x_{(p,i)}-x_{(p,i+1)})\\
&\geq\sum_{p=1}^{l}\min_{i}a_{(p,i,i+1)}(x_\text{max}-x_\text{min})\geq a_\text{min}^{(l)}v(\bm x)
\end{aligned}
\ene
where $a_{(p,i,i+1)}$ is the weight of the edge connecting nodes $x_{(p,i)}$ and $x_{(p,i+1)}$.

Letting $\tilde{x}=\frac{1}{2}(\max_i(x_i)+\min_i(x_i))$, we have
\bee\label{eq7_theo3}
\begin{aligned}
&\|\bm x-\bar{x}\bone\|\|\nabla g(\bar{x}\bone)\|\leq\|\bm x-\tilde{x}\bone\|\|\nabla g(\bar{x}\bone)\|\\
&\leq\sqrt{n}\|\bm x-\tilde{x}\bone\|_\infty\cdot\sqrt{n}\|\nabla g(\bar{x}\bone)\|_\infty\leq \frac{n c}{2}v(\bm x).
\end{aligned}
\ene
where the first inequality follows from the fact that $\bar{x}$ minimizes $\|\bm x-\alpha\bone\|$ with respect to (w.r.t.) $\alpha$ for all $\bm x$. Eqs. \eqref{eq1_theo3}, \eqref{eq2_theo3} and \eqref{eq7_theo3} jointly imply the following inequality 
\bee\label{eq6_theo3}
\begin{aligned}
&\tilde{f}_\lambda(\bm x)-f^\star\geq f(\bar{x})-f^\star+(\lambda a_\text{min}^{(l)}-\frac{ cn}{2})v(\bm x).
\end{aligned}
\ene

Since $\lambda>{n c}/({2a_\text{min}^{(l)}}),v(\bm x)\geq0,\ \forall \bm x\in\bR^n$ and $f(\bar{x})\geq f^\star, \forall \bar{x}\in\bR$, then the right hand side of \eqref{eq6_theo3} is nonnegative. That is, $\tilde{f}_\lambda(\bm x)\geq f^\star$ for all $\bm x \in\bR^n$. 

Moreover, it follows from \eqref{penalty} that $\tilde{f}_\lambda(x^\star\bone)=f^\star$ for any $x^\star\in \cX^\star$, i.e., $\tilde{f}_\lambda(\bm x)=f^\star$ for any $\bm x \in \tilde{\cX}^\star$.  It remains to show that $\tilde{f}_\lambda(\bm x)> f^\star$ for all $\bm x\notin \tilde{\cX}^\star$, which includes:
\begin{enumerate}[label=Case (\alph*):,leftmargin =50pt]
\item $\bm x\neq \alpha\bone$ for any $\alpha\in \bR$,
\item $\bm x= \alpha\bone$ for some $\alpha\notin \cX^\star$.
\end{enumerate}

For Case (a), $v(\bm x)$ is strictly positive, and hence we know that $\tilde{f}_\lambda(\bm x)> f^\star$ from \eqref{eq6_theo3}. For Case (b), we have $v(\bm x)=0$. By \eqref{eq6_theo3} we have that $\tilde{f}_\lambda(\bm x)\geq f(\bar{x})=f(\alpha)>f^\star$. Thus, $\tilde{f}_\lambda(\bm x)> f^\star$ for all $\bm x\notin \tilde{\cX}^\star$, which completes the proof of part (a). The proof of part (b) is given in Appendix \ref{appendix_a}.
\end{proof}

It is worth mentioning that \eqref{lowerbound} in Theorem \ref{theo1} also holds for the multi-dimension case if Assumption \ref{assum2}\ref{assum2a} is replaced with $\|\nabla f_i({\bm x})\|\leq c$ for all $i$ and ${\bm x}$.

Theorem \ref{theo1} provides a sufficient condition for the equivalence between problems \eqref{reformulate} and \eqref{penalty}, and allows us to focus only on problem $\eqref{penalty}$. This result is nontrivial even though the penalty method has been widely studied in the  literature on  optimization theory \cite{bertsekas2015convex,bazaraa2013nonlinear}. By \cite{bazaraa2013nonlinear}, a lower bound for $\lambda$ can be selected as the largest absolute value of Lagrange multipliers of the equality constraints in \eqref{reformulate}. However, a Lagrange multiplier usually cannot be obtained before solving a dual problem, and it is unclear how to establish the relationship between the Lagrange multiplier and the network structure. Via a different technique, Theorem \ref{theo1} provides an explicit lower bound for $\lambda$ in terms of the network size and its connectivity, and is tighter than the bounds in \cite{franceschelli2017finite} and \cite{lin2017distributed}.

In fact, the lower bound in Theorem \ref{theo1} can be tight in some cases as shown in the following Example \ref{example_1} and Section \ref{sec_sim_a}. Note that a too large $\lambda$ may have negative effects on the transient performance of \ref{protocol}, as we will demonstrate later in Section \ref{sec_sim_a}. Thus, the tighter bound in Theorem \ref{theo1} allows us to choose a smaller $\lambda$ in applications.

\begin{example}\label{example_1}
Consider the graph in Fig. \ref{fig1b:sec_3} with unit edge weights, i.e., $a_{ij}=1$ for all $(i,j)\in\cV$. Let $f_1(x)=|x|,f_2(x)=|x-2|,f_3(x)=|x-4|,f_4(x)=|x-6|$ and $f(x)=\sum_{i=1}^4f_i(x)$. It is not difficult to compute that the optimal value of $f(x)$ is 8 and the set of optimal solutions is a closed interval $[2, 4]$.  By \eqref{penalty}, the corresponding penalized problem is given as
\[
\begin{aligned}
\tilde{f}_\lambda(\bm x)&=|x_1|+|x_2-2|+|x_3-4|+|x_4-6|+\\
&~~~\lambda(|x_1-x_2|+|x_2-x_3|+|x_3-x_4|+|x_4-x_1|).
\end{aligned}
\]
 Theorem \ref{theo1} implies that $\tilde{f}_\lambda(\bm x)$ has the same optimal value as $f(x)$ and the set of optimal solutions is $\tilde{\cX}^\star=\{x^{\star} \bone|x^{\star}\in [2, 4]\}$, provided that $\lambda>{4\cdot 1}/({2\cdot 2})=1$. 

Given any $\lambda\leq 1$, consider $\bm x =[2,2,4,4]^\mathsf{T}\notin \tilde{\cX}^\star$. Clearly, $\tilde{f}_\lambda(\bm x)=4+4\lambda\leq f^\star=8$, which implies that the set of optimal solutions of the penalized problem is not $\tilde{\cX}^\star$. Thus for any $\lambda\leq 1$,  the original problem $f(x)$ cannot be solved via the penalized problem $\tilde{f}_\lambda(\bm x)$, and the lower bound in \eqref{lowerbound} is tight in this example. \qed
\end{example}

\begin{figure}[!t]
\centering
\subfloat[]{\label{fig1a:sec_3}{\includegraphics[width=0.33\linewidth]{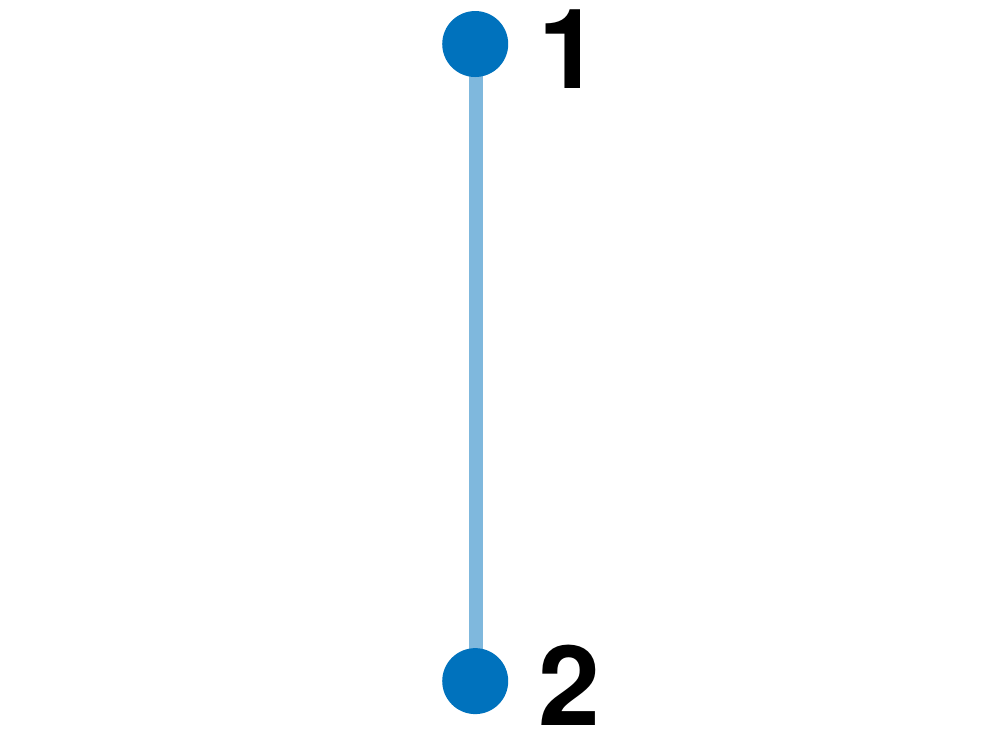}}}
\subfloat[]{\label{fig1b:sec_3}\includegraphics[width=0.33\linewidth]{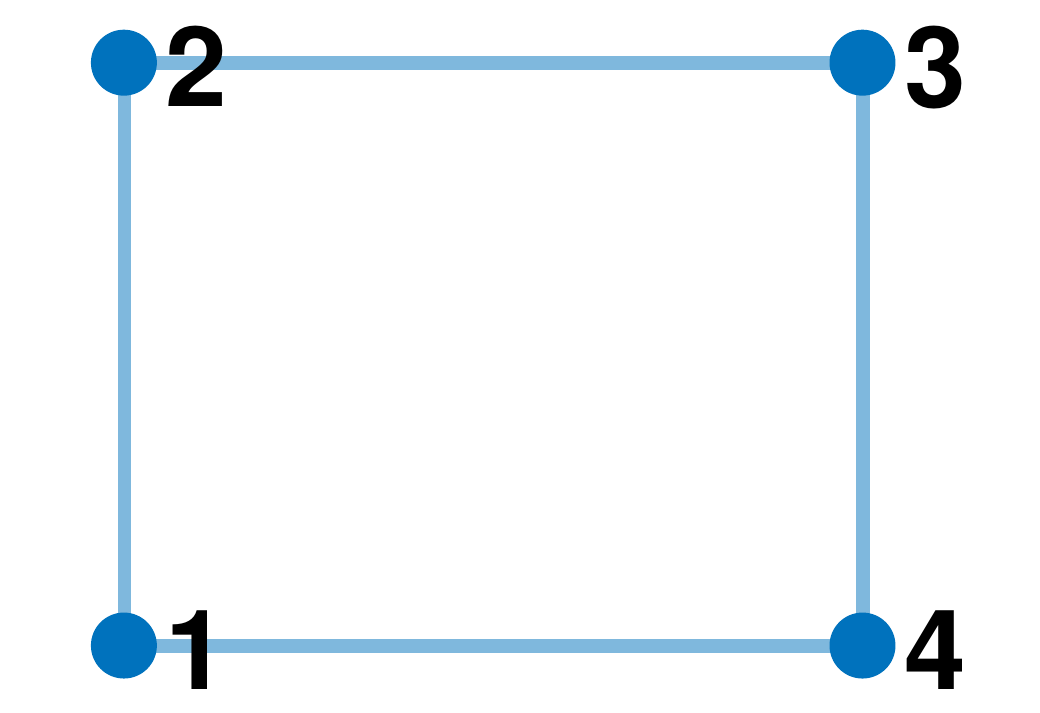}}
\subfloat[]{\label{fig1c:sec_3}\includegraphics[width=0.33\linewidth]{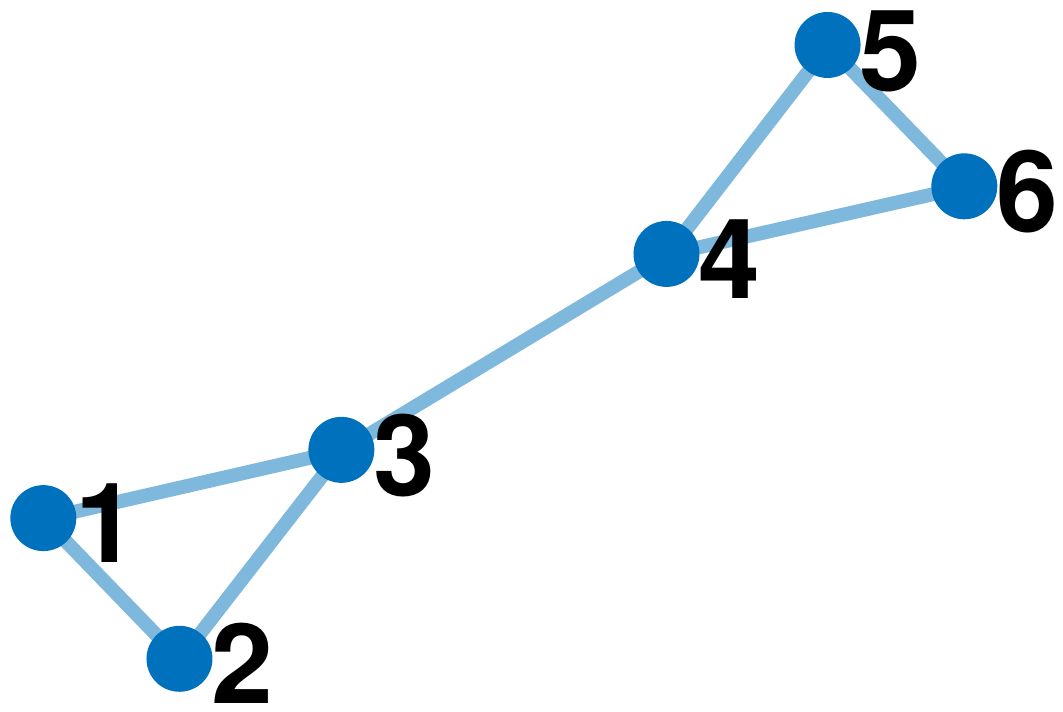}}
\caption{Some graphs.}
\label{fig1:sec_3}
\end{figure}

The lower bound in \eqref{lowerbound} is in a simple form and $a_{\min}^{(l)}$ cannot be easily replaced. One may consider to use the minimum degree of the network, i.e., $d_{m}=\min_{i\in\cV} \sum_{j=1}^n a_{ij}$. This is impossible in some cases. Consider the $1$-connected graph in Fig. \ref{fig1c:sec_3} with unit edge weights. Then, $a_{\min}^{(1)}=1$ and $d_{m}=2$.  Let $[s_1,...,s_6]=[1,2,3,4,5,6]$ and $f_i(x)=|x-s_i|,\ \forall i\in\{1,...,6\}$. Set $\bm x =[x_1,...,x_6]^\mathsf{T}=[3,3,3,4,4,4]^\mathsf{T}$ and use similar arguments as Example \ref{example_1}, one can inference that the lower bound $\underline{\lambda}$ in \eqref{lowerbound} cannot be reduced to ${n c}/(2d_{m})=3/2$.

A similar penalty method interpretation of \eqref{dgd} with constant $\rho^k$ is provided in \cite{mokhtari2017network}, where the penalty function is chosen as $\bm x^\mathsf{T}L\bm x=\frac{1}{2}\sum_{i,j} a_{ij}(x_i-x_j)^2$ and $L$ is the graph Laplacian matrix. However, such a quadratic penalty function cannot always guarantee the existence of a finite $\lambda$ for the equivalence of the two problems. 

%

By \cite{bertsekas2015convex}, ${\bm x^\star}$ is an optimal solution of \eqref{penalty} if and only if $0\in \partial \tilde{f}_\lambda({\bm x^\star})$. Part (b) of Theorem \ref{theo1} shows that for any $\bm x \notin\{\alpha\bone| \alpha\in\bR\}$, the norm of the corresponding subgradient is uniformly greater than a positive lower bound, which clearly shows the non-optimality of ${\bm x}$.

Assumption \ref{assum2}\ref{assum2a} in Theorem \ref{theo1} can also be  removed to ensure the equivalence of the problems \eqref{original} and \eqref{penalty}. 
\begin{theo}\label{theo8} Suppose that Assumption \ref{assum1} holds, and that the multi-agent network is $l$-connected. If the penalty factor is selected as
\bee\label{eq2_theo8}
\lambda>\frac{nc^\star}{2a_\text{min}^{(l)}},
\ene
where $c^\star=\min_{x^\star\in\cX^\star}\max_{y\in\cup_{i\in\cV}\partial f_i(x^\star)} |y|$, and $a_\text{min}^{(l)}$ is given in \eqref{minsum}, then the set of optimal solutions and optimal value of the penalized problem \eqref{penalty} are $ \tilde{\cX}^\star=\{x^\star\bone|x^\star\in \cX^\star\}$ and $f^\star$, respectively.
\end{theo}
\begin{proof}
We prove the results only for differentiable $f_i(x)$ to save space, where $\nabla f_i(x)$ now becomes the gradient. Similar ideas can also be applied to the non-differentiable case. We first show that for any $x_a^\star, x_b^\star\in\cX^\star$, $\nabla f_i(x_a^\star)=\nabla f_i(x_b^\star), \forall i\in\cV$. There is no loss of generality to let $x_a^\star<x_b^\star$. Since $f_i(x)$ is convex,
\[
\begin{aligned}
f_i(x_a^\star)\geq f_i(x_b^\star)+\nabla f_i(x_b^\star)(x_a^\star-x_b^\star),\\
f_i(x_b^\star)\geq f_i(x_a^\star)+\nabla f_i(x_a^\star)(x_b^\star-x_a^\star).
\end{aligned}
\] 

Summing the two inequalities leads to $(x_a^\star-x_b^\star)(\nabla f_i(x_a^\star)-\nabla f_i(x_b^\star))\geq 0$, which  implies $\nabla f_i(x_a^\star)\leq\nabla f_i(x_b^\star)$. This together with that $\sum_{i=1}^n\nabla f_i(x_a^\star)=\sum_{i=1}^n\nabla f_i(x_b^\star)=0$ yields that $\nabla f_i(x_a^\star)=\nabla f_i(x_b^\star)$. Then given an arbitrary $\epsilon>0$, we can let $d_i^\epsilon =|\nabla f_i(x^\star)|+\epsilon$ for all $x^\star\in\cX^\star$. 

Denote the conjugate function of $f_i(x)$ by $f_i^\star(y)$, i.e., $f_i^\star(y)=\sup_{x\in\bR} \{xy-f_i(x)\}$. Consider the following optimization problem 
\bee\label{eq1_theo8}
\minimize_{x\in\bR}\ f_\epsilon(x):=\sum_{i=1}^n f_i^\epsilon(x)
\ene
where 
\bee\label{defif}
f_i^\epsilon(x)=\left\{\begin{aligned}
        &d_i^\epsilon x-f_i^\star(d_i^\epsilon ),&&\text{if }\nabla f_i(x)\geq d_i^\epsilon,\\
	&-d_i^\epsilon x-f_i^\star(-d_i^\epsilon ),&&\text{if }\nabla f_i(x)\leq -d_i^\epsilon,\\ 
	&f_i(x),&&\text{otherwise}.\end{aligned}\right.
\ene
We claim the following. 

{\em Claim 1}: Problem \eqref{eq1_theo8} is equivalent to the problem \eqref{original}.

To this end, we note that $f_\epsilon(x)$ is convex as well. Since $\nabla f_i(x^\star)\in(-d_i^\epsilon, d_i^\epsilon )$ for all $x^\star\in\cX^\star$, it follows that
\bee\label{eq1_theo9}
\nabla f_\epsilon(x^\star)=\sum_{i=1}^n\nabla f_i(x^\star)=0,\ \forall x^\star\in\cX^\star. 
\ene

For any $x\notin\cX^\star$, where there is no loss of generality to assume that $x$ is strictly greater than all elements in $\cX^\star$, we obtain that $\nabla f_i(x)\geq\nabla f_i(x^\star)$ as $(x-x^\star)(\nabla f_i(x)-\nabla f_i(x^\star))\geq 0$.
Jointly with \eqref{defif}, it implies that  for all $x^\star\in\cX^\star$, 
 \bee\label{inequal}
 \nabla f_i^\epsilon(x)=\min\{\nabla f_i(x),d_i^\epsilon\}\geq\nabla f_i(x^\star), \forall i\in\cV.
 \ene 
 If the inequalities in \eqref{inequal} strictly hold for some $i\in\cV$, then
$$\nabla f_\epsilon(x)=\sum_{i=1}^n\nabla f_i^\epsilon(x)>\sum_{i=1}^n\nabla f_i(x^\star)=0.$$
That is, $\nabla f_\epsilon(x)\neq 0$ for any $x\notin\cX^\star$. While $f_\epsilon(x^\star)=f(x^\star)$ for any $x^\star\in\cX^\star$, Claim 1 is verified.

It remains to show that the inequalities in \eqref{inequal} must strictly hold for some $i\in\cV$. On the contrary,  suppose that $\nabla f_i^\epsilon(x)=\nabla f_i(x^\star)$ for all $i\in\cV$. Since $d_i^\epsilon>\nabla f_i(x^\star)$ for all $i\in\cV$, it follows from \eqref{inequal} that $\nabla f_i(x)=\nabla f_i(x^\star)$ for all $i\in\cV$. Then, $\sum_{i=1}^n\nabla f_i(x)=\sum_{i=1}^n\nabla f_i(x^\star)=0$, which contradicts that $x\notin\cX^\star$. 

{\em Claim 2}:  Problem \eqref{eq1_theo8} is equivalent to  the penalized problem  \eqref{penalty} in the same sense as Theorem \ref{theo1}(a).

Consider the following penalized problem
\bee\label{penalizedprob}
\minimize_{\bm x\in\bR^n}\ \tilde{f}_{\epsilon,\lambda}(\bm x):=g_\epsilon(\bm x)+\lambda h(\bm x)
\ene
where $g_\epsilon(\bm x):=\sum_{i=1}^{n} f_i^\epsilon(x_i)$.

Noting that $|\nabla f_i^\epsilon(x)|\leq\max_{i\in\cV}{d_i^\epsilon}:=d_M^\epsilon$ for all $x$, it follows from Theorem \ref{theo1}  that by selecting $\lambda>{n d_M^\epsilon}/(2a_\text{min}^{(l)})$, problem \eqref{eq1_theo8} is equivalent to problem \eqref{penalizedprob}.

Since $f_i^\epsilon(x)\leq f_i(x)$ for all $x\in\bR$ and $i\in\cV$, we have that $g_\epsilon(\bm x)\leq g(\bm x)$. Then, $\tilde{f}_{\epsilon,\lambda}(\bm x)\leq\tilde{f}_{\lambda}(\bm x)$ for all $\bm x\in\bR^n$, and $\tilde{f}_{\lambda}(\bm x)\geq\tilde{f}_{\epsilon,\lambda}(\bm x)\geq\min_{x\in\bR}f_{\epsilon}(x)=f^\star$, where all equalities hold if and only if $\bm x\in\tilde{\cX}^\star$. That is, the penalized problem  \eqref{penalty} is equivalent to problem \eqref{penalizedprob}, which implies Claim 2 as well. 

Finally, we conclude that the penalized problem \eqref{penalty} is equivalent to the original problem \eqref{original}, provided that $\lambda>{nd_M^\epsilon}/({2a_\text{min}^{(l)}})$. Given any $\lambda>{nc^\star}/({2a_\text{min}^{(l)}})$, there exists a positive $\epsilon>0$ such that $\lambda>{nd_M^\epsilon}/({2a_\text{min}^{(l)}})$, which completes the proof.
\end{proof}
\begin{remark}
It is usually difficult to obtain $c^\star$ in \eqref{eq2_theo8}. In applications, an upper bound can be used instead. Specifically, let $x_i^\text{opt}$ be an optimal solution  of $f_i(x)$, then we have $c^\star\leq \min_i\max_{j}|\nabla f_i(x_j^\text{opt})|$.
\end{remark}

Using the novel idea of constructing the optimization problem \eqref{eq1_theo8}, Theorem \ref{theo8} extends the results of Theorem 1 to objective functions with unbounded (sub)gradients, which includes quadratic functions as a special case. Obviously, the quadratic form constitutes an important class of objective functions in real applications.  

\section{Convergence Analysis}\label{sec_4}
In this section we examine the convergence behavior of \ref{protocol}. If $\rho^k$ is diminishing, all agents converge to the same optimal solution of problem \eqref{original} under \ref{protocol}. With a constant stepsize, all agents eventually converge to a neighborhood of an optimal solution. For both cases, we perform the non-asymptotic analysis to determine the their convergence rates. 

Let $\{\bm x^k\}$ be generated by \eqref{subgrad}, it follows from  \cite{bertsekas2015convex} to easily establish the following inequalities.
\begin{enumerate}[label=(\alph*)]
\item Under Assumption \ref{assum2}\ref{assum2a}, it holds that for all $x^\star\in\cX^\star$,
\bee\label{eq_sg}
\begin{aligned}
&\|\bm x^{k+1}-x^\star\bone\|^2\\
&\leq\|\bm x^{k}-x^\star\bone\|^2-2\rho^k(\tilde{f}_\lambda(\bm x^k)-f^\star)+(\rho^k)^2c_a^2\\
\end{aligned}
\ene
\item Under Assumption \ref{assum2}\ref{assum2b}, it holds that for all $x^\star\in\cX^\star$,
\bee\label{eq_sg2}
\hspace{-0.2cm}\begin{aligned}
\|\bm x^{k+1}-x^\star\bone\|^2&\leq(1+(\rho^k)^2 c^2)\|\bm x^{k}-x^\star\bone\|^2\\
&\quad-2\rho^k(\tilde{f}_\lambda(\bm x^k)-f^\star)+(\rho^k)^2c_b^2
\end{aligned}
\ene
\end{enumerate}
where $c_a,c_b$ are given in \eqref{c_a} and \eqref{c_b}, respectively. 

Proof of the convergence of \ref{protocol} with diminishing stepsizes is straightforward.
\begin{theo}\label{theo2}  Suppose that the conditions in Theorem \ref{theo1}, or the conditions in Theorem \ref{theo8} and Assumption \ref{assum2}\ref{assum2b} hold. Let $\{\bm x^k\}$ be generated by \ref{protocol} and $\rho^k$ satisfy
\beq
\sum_{k=0}^\infty\rho^k=\infty,~\text{and}~ \sum_{k=0}^\infty(\rho^k)^2<\infty.
\enq
Then, there is some $x^\star\in\cX^\star$ such that
$\lim_{k\ra\infty}\bm x^k=x^\star\bone.$
\end{theo}
\begin{proof}
Recall that \ref{protocol} is the exact iteration of the subgradient method of problem \eqref{penalty}. It follows from Proposition 3.2.6  in \cite{bertsekas2015convex} that $\{\bm x^k\}$ converges to some optimal solution of problem \eqref{penalty}. Combined with Theorem \ref{theo1} or Theorem \ref{theo8}, the result follows immediately.
\end{proof}
\begin{remark} In the sequel, we only present the results for the case of bounded subgradient, i.e., Assumption \ref{assum2}\ref{assum2a} holds, which can be replaced by Assumption \ref{assum2}\ref{assum2b}. In fact, we only need to use \eqref{eq_sg2} to replace \eqref{eq_sg} to establish the main results. Due to space limitation, we do not include details here. 
\end{remark}

Our next result provides the non-asymptotic result to evaluate the convergence rate for $\rho^k=k^{-\alpha},\alpha\in[0.5,1]$. To this end, we define
\bee\label{defidis}
d(\bm x)=\min_{x^\star\in \cX^\star}\|\bm x-x^\star\bone\|.
\ene

\begin{theo}\label{theo3}
Suppose that the conditions in Theorem \ref{theo1} hold, and let $\{\bm x^k\}$ be generated by \ref{protocol}. If $\rho^k={k^{-\alpha}}, \alpha\in(0.5,1]$, then
\bee\label{eq5_theo3}
 \begin{aligned}
\min_{1< t\leq {k}}f(x_i^t)-f^\star&\leq\frac{(2\alpha-1)d(\bm x^0)^2+2\alpha c_a^2}{2(2\alpha-1)s(k)},\ \forall i\in\cV\\
\end{aligned}
\ene
where $\bm x^0$ is the initial point, and
\[
s({k})=\left\{\begin{aligned}&\frac{1}{1-\alpha}(k^{1-\alpha}-1),&&\text{ if }\alpha\in(0.5,1),\\ &\ln(k),&&\text{ if }\alpha=1.\end{aligned}\right.
\]
Moreover, if $\rho^k=1/\sqrt{k}$,  we have that
\begin{align}
\min_{1< t\leq {k}}f(x_i^k)-f^\star&\leq\frac{d(\bm x^0)^2+c_a^2\ln(k)}{4\sqrt{k}},\ \forall i\in\cV.
\end{align}
\end{theo}
\begin{proof}
By Theorem \ref{theo2}, $\{\bm x^k\}$ is a convergent sequence. For any $x^\star\in\cX^\star$, it follows from \eqref{eq_sg} that
\bee\nonumber
2\rho^t(\tilde{f}_\lambda(\bm x^t)-f^\star)\leq\|\bm x^{t}-x^\star\bone\|^2-\|\bm x^{t+1}-x^\star\bone\|^2+(\rho^t)^2c_a^2.
\ene
Summing the above relation over $t\in\{1,...,k\}$ yields
\begin{align}
&2\sum_{t=1}^{k}\rho^t(\tilde{f}_\lambda(\bm x^t)-f^\star)\\
&\leq\|\bm x^{0}-x^\star\bone\|^2-\|\bm x^{{k}+1}-x^\star\bone\|^2+\sum_{t=1}^{k}(\rho^t)^2c_a^2\\
&\leq d(\bm x^0)^2+\sum_{t=1}^{k}(\rho^t)^2c_a^2
\end{align}
where the last inequality holds by choosing $x^\star=\argmin_{x\in\cX^\star}\|\bm x^0-x\bone\|$. Then, it follows that
\bee\label{eq4_theo3}
\min_{0\leq t\leq {k}}\tilde{f}_\lambda(\bm x^t)-f^\star\leq\frac{d(\bm x^0)^2+\sum_{t=1}^{k}(\rho^t)^2c_a^2}{2\sum_{t=1}^{k}\rho^t}.
\ene

Since 
$\int_1^{k}\frac{1}{x^\alpha}dx<\sum_{t=1}^{k} \frac{1}{t^\alpha}<\int_1^{k}\frac{1}{x^\alpha}dx+1,$
we have that $\sum_{t=1}^{k}(\rho^t)^2<\int_1^{k}\frac{1}{x^{2\alpha}}dx+1=\frac{1-k^{1-2\alpha}}{2\alpha-1}+1<\frac{2\alpha}{2\alpha-1}$, and $\sum_{t=1}^{k}\rho^t>\int_1^{k}\frac{1}{x^{\alpha}}dx=s(k)$ if $\alpha\in(0.5,1]$. Together with \eqref{eq4_theo3}, this implies 
\bee\label{eq3_theo3}
\min_{0\leq t\leq {k}}\tilde{f}_\lambda(\bm x^t)-f^\star\leq\frac{(2\alpha-1)d(\bm x^0)^2+2\alpha c_a^2}{2(2\alpha-1)s(k)}.
\ene

By \eqref{eq1_theo3}, it follows that 
\begin{align}
 f({x_i^t}) & \leq\tilde{f}_\lambda(\bm x^t)-\lambda h(\bm x^t)+\|\bm x^t-x_i^t\bone\|\cdot\|\nabla g(x_i^t\bone)\|
\end{align}
which combined with
\begin{align}
 &\|\bm x^t-x_i^t\bone\|\cdot\|\nabla g(x_i^t\bone)\|\\
 & \leq\sqrt{n}\|\bm x^t-x_i^t\bone\|_\infty\cdot\sqrt{n}\|\nabla g(x_i^t\bone)\|_\infty\leq ncv(\bm x^t)
\end{align}
and \eqref{eq5_theo1} yields that
\begin{align}\label{eq8_theo3}
 f({x_i^t}) & \leq\tilde{f}_\lambda(\bm x^t)-\lambda h(\bm x^t)+\frac{nc}{a_\text{min}^{(l)}}h(\bm x^t)\\
 &=g(\bm x^t)+\frac{nc}{a_\text{min}^{(l)}}h(\bm x^t)\leq \tilde{f}_{2\lambda}(\bm x^t)
\end{align}
where the last inequality follows from $\lambda>nc/(2a_{\min}^{(l)})$. 

In view of \eqref{eq3_theo3}, the above implies
\begin{align}\label{eq9_theo3}
 \min_{0\leq t\leq {k}}\tilde{f}_{2\lambda}(\bm x^t)-f^\star\leq\frac{(2\alpha-1)d(\bm x^0)^2+2\alpha c_a^2}{2(2\alpha-1)s(k)}.
\end{align}
The result for $\alpha\in(0.5,1)$ follows from \eqref{eq8_theo3} and \eqref{eq9_theo3}, while the result for $\alpha=0.5$ is from \eqref{eq4_theo3} and \eqref{eq8_theo3}.
\end{proof}

Theorem \ref{theo3} reveals that the convergence rate of the objective function lies between $O(1/\ln(k))$ and $O(\ln(k)/\sqrt{k})$, depending on the choice of $\rho^k$. If $f(x)$ is non-differentiable, the convergence rate is essentially of the same with that of the classical distributed algorithm \eqref{dgd} \cite{shi2015extra}. Thus using only the sign of relative state does not lead to reduction in the convergence rate. However, if $f(x)$ is differentiable or strongly convex,  \ref{protocol} may converge at a rate slower than that of \eqref{dgd} due to the non-smoothness of the second term in \ref{protocol}. Harnessing smoothness to accelerate distributed optimization has been well studied; see e.g., \cite{qu2017harnessing}.

For a constant stepsize, \ref{protocol} approaches a neighborhood of an optimal solution as fast as $O(1/k)$ and the error is proportional to the stepsize. These results are formally stated in Theorem \ref{theo4} and Theorem \ref{theo5}.

\begin{theo}\label{theo4} Suppose that the conditions in Theorem \ref{theo1} hold, and let $\{\bm x^k\}$ be generated by \ref{protocol}. If $\rho^k=\rho$, then
\bee\label{eq7_theo4}
\begin{aligned}
&\limsup_{k\rightarrow\infty}d(\bm x^k)\leq 2\sqrt{n}\max\left\{\tilde{d}(\rho),\frac{\rho c_a^2}{2\lambda a_\text{min}^{(l)}- cn}\right\}+\rho c_a
\end{aligned}
\ene
where $\tilde{\cX}(\rho)=\{x|f(x)\leq f^\star+{\rho c_a^2}/{2}\}$ and $\tilde{d}(\rho)=\max_{x\in\tilde{\cX}(\rho)}d(x)<\infty$.
\end{theo}
\begin{proof}
See Appendix \ref{appendix_b}. 
\end{proof}

In Theorem \ref{theo4}, $\tilde{d}(0)=0$ and $\tilde{d}(\rho)$ is increasing in $\rho$. Thus, \ref{protocol} under a constant stepsize finally approaches a neighborhood of $x^\star\bone$ for some $x^\star\in\cX^\star$, the size of which decreases to zero as $\rho$ tends to zero. If the order of growth of $f$ near the set of optimal solutions  is available, then $\tilde{d}(\rho)$ can even be determined explicitly, which is given in Corollary \ref{coro1}.
\begin{coro}\label{coro1}
Suppose that the conditions in Theorem \ref{theo4} hold, and that $f(x)$ satisfies
\[f(x)-f^\star\geq\gamma (d(x))^\alpha\]
where $\gamma>0$ and $\alpha\geq 1$. Then
\bee\nonumber
\limsup_{k\rightarrow\infty}d(\bm x^k)\leq 2\sqrt{n}\max\left\{\left(\frac{\rho c_a^2}{2\gamma}\right)^\frac{1}{\alpha},\frac{\rho c_a^2}{2\lambda a_\text{min}^{(l)}- cn}\right\}+\rho c_a
\ene
\end{coro}
\begin{proof}
Noting that $\tilde{d}(\rho)\leq({\rho c_a^2}/{2\gamma})^\frac{1}{\alpha}$, the result follows directly from Theorem \ref{theo4}.
\end{proof}

The following theorem evaluates the convergence rate when the stepsize is a constant.

\begin{theo}\label{theo5}
Suppose that the conditions in Theorem \ref{theo4} hold. Then
\bee\label{eq1_theo4}
\begin{aligned}
\min_{0\leq t\leq k}f(x_i^t)-f^\star&\leq\frac{\rho c_a^2}{2}+\frac{d(\bm x^0)^2}{2\rho k},\ \forall i\in\cV.\\
\end{aligned}
\ene
\end{theo}
\begin{proof}
From \eqref{eq4_theo3} we know that
\bee\nonumber
\min_{0\leq t\leq {k}}\tilde{f}_\lambda(\bm x^t)-f^\star\leq\frac{d(\bm x^0)^2+k\rho^2c_a^2}{2\rho k},
\ene
which together with \eqref{eq8_theo3} implies the result.
\end{proof}
\begin{remark}
The following conclusions can be easily arrived at from Theorem \ref{theo5}.
\begin{enumerate}[label=(\alph*)]
\item   $\min_{0\leq t\leq k}f(x_i^t)$ approaches the interval $[f^\star,f^\star+\frac{\rho c_a^2}{2}]$ at a rate of $O(1/k)$.
\item Given $k$ iterations, let $\rho=\frac{1}{c_a}\frac{d(\bm x^0)}{\sqrt{k}}$, which minimizes the right-hand-side of \eqref{eq1_theo4}. Then
\bee\nonumber
\begin{aligned}
\min_{0\leq t\leq k}f(x_i^t)-f^\star&\leq c_a\frac{d(\bm x^0)}{\sqrt{k}},\ \forall i\in\cV.\\
\end{aligned}
\ene
The multi-agent network converges only to a neighborhood of an optimal solution with an error size $O(k^{-1/2})$. 
\end{enumerate}
\end{remark}

\section{A Distributed Algorithm with Noisy Relative State Information}\label{sec_7}
In real applications, the measurement of the relative state may be noise corrupted. This happens because of including inaccurate sensors, unreliable communications, and poor sensing environment. To capture such inaccuracies, we replace $\text{sgn}(x_i^k-x_j^k)$ in \ref{protocol} with $\text{sgn}(x_i^k-x_j^k+\epsilon_{ij}^k)$, where, for each $i,j\in\cV$, $\{\epsilon_{ij}^k\}$ is a sequence of independent and identically distributed (i.i.d.) Gaussian random variables with zero mean and variance $\sigma_{ij}$, i.e., $\epsilon_{ij}^k\sim N(0,\sigma_{ij}^2)$ . Our objective is then to study the following algorithm
\begin{tcolorbox}[ams equation, size=small]\label{protocol_noisy}\tag{Algo. 2}
x_i^{k+1}=x_i^{k}+\lambda\rho^k\sum_{j\in \cN_i}a_{ij}\text{sgn}(x_j^k-x_i^k+\epsilon_{ij}^k)-\rho^k\nabla f_i(x_i^k).
\end{tcolorbox}
\noeqref{protocol_noisy}
Then, \ref{protocol_noisy} is exactly the iteration of the stochastic subgradient method of the following penalized problem
\bee\label{noisy}
\minimize_{\bm x\in\bR^n}\ \bar{f}_\lambda(\bm x):=g(\bm x)+\lambda \bar{h}(\bm x)
\ene
where $g(x)$ is given in \eqref{reformulate}, and
\bee\label{hx_noisy}
\bar{h}(\bm x)=\frac{1}{2}\sum_{i=1}^n \sum_{j\in \cN_i}a_{ij}\bE\{|x_i-x_j+\epsilon_{ij}|\},
\ene
and $\bE(x)$ denotes the expectation of a random variable $x$. Then, we have the following result.  
%

\begin{lemma}\label{lemma4}
Let $\{\bm x^k\}$ be generated by \ref{protocol_noisy} and $\rho^k$ satisfy
\beq
\sum_{k=0}^\infty\rho^k=\infty,~\text{and}~ \sum_{k=0}^\infty(\rho^k)^2<\infty.
\enq
Under Assumptions \ref{assum1} and \ref{assum2}\ref{assum2b},  $\{\bm x^k\}$ converges to some optimal solution of the problem \eqref{noisy}
\end{lemma}
\begin{proof}
 Since \ref{protocol_noisy} is exactly the stochastic subgradient method of the problem \eqref{noisy}, the result follows directly from the convergence theorem of stochastic subgradient methods. See e.g. \cite[Chapter 5]{borkar2008stochastic}.
\end{proof}

Note that Assumption \ref{assum2}\ref{assum2b} is weaker than Assumption \ref{assum2}\ref{assum2a}. The study of convergence rate is much more involved, see e.g. \cite{lim2011convergence} where more technical assumptions are needed on the objective function. As our focus is not on the convergence rate of stochastic subgradient methods and also due to the space limitation, we do not discuss here the convergence rate of \ref{protocol_noisy}  and leave it for future work.

Due to the presence of noise, we would not expect problem \eqref{penalty} and problem \eqref{noisy} to be equivalent. However, we can still evaluate the difference between their optimal solutions. To this end, we introduce the folded normal distribution below.
\begin{lemma}[Folded Normal Distribution,\cite{leone1961folded}]\label{fnormal}
If $x\sim N(\mu,\sigma^2)$, then $y=|x|$ has a folded normal distribution with parameters $\mu$ and $\sigma^2$, and
\bee\nonumber
\bE(y)=\mu[1-2\Phi(-\frac{\mu}{\sigma})]+\sigma\sqrt{\frac{2}{\pi}}\exp\left(-\frac{\mu^2}{2\sigma^2}\right)
\ene
where $\Phi(\cdot)$ is the standard normal cumulative distribution function. In particular, if $\mu=0$, then $\bE(y)=\sigma\sqrt{{2}/{\pi}}$.
\end{lemma}
\begin{theo}\label{theo7}
Suppose that the conditions in Theorem \ref{theo1} hold and let $\bar{\bm x}^\star=[\bar{x}_1^\star,...,\bar{x}_n^\star]^\mathsf{T}$ be an optimal solution of problem \eqref{noisy}. Then, there exists $x^\star\in \cX^\star$ such that 
\bee\nonumber
\|{\bar{\bm x}^\star}-x^\star \bone\|_{\infty}\le v(\bar{\bm x}^\star)\le \sqrt{\frac{2}{\pi}}\frac{2\lambda \sigma_s}{2\lambda a_\text{min}^{(l)}- cn}
\ene
where $\sigma_s=\frac{1}{2}\sum_{i,j}a_{ij}\sigma_{ij}$.
\end{theo}
\begin{proof}
Since $\bar{\bm x}^\star$ is an optimal solution of \eqref{noisy}, we have
\bee\nonumber
\begin{aligned}
\bar{f}_\lambda(\bar{\bm x}^\star)&=g(\bar{\bm x}^\star)+\frac{\lambda}{2} \sum_{i,j}a_{ij}\bE|\bar{x}_i^\star-\bar{x}_j^\star+\epsilon_{ij}|\\
&\leq \bar{f}_\lambda(x^\star\bone)=f^\star+\frac{\lambda}{2} \sum_{i,j}a_{ij}\bE|\epsilon_{ij}|.
\end{aligned}
\ene
Thus
$
g(\bar{\bm x}^\star)-f^\star\leq\frac{\lambda}{2} \sum_{i,j}a_{ij}(\bE|\epsilon_{ij}|-\bE|\bar{x}_i^\star-\bar{x}_j^\star+\epsilon_{ij}|).
$

Let $\mu_{ij}=\bar{x}_i-\bar{x}_j$ and $t(\mu)=\mu[2\Phi(-\frac{\mu}{\sigma_{ij}})-1]+\sigma_{ij}\sqrt{\frac{2}{\pi}}(1-\exp(-\frac{\mu^2}{2\sigma_{ij}^2}))$. It follows from Lemma \ref{fnormal} that $ \bE|\epsilon_{ij}|-\bE|\bar{x}_i-\bar{x}_j+\epsilon_{ij}|=t(\mu_{ij})$. 

Now we show that $t(\mu)<-|\mu|+\sigma_{ij}\sqrt{{2}/{\pi}}$ for all $\mu\in \bR$. Let $\bar{t}(\mu):=t(\mu)+|\mu|-\sigma_{ij}\sqrt{{2}/{\pi}}$. For any $\mu\neq 0$, we have that
\bee\nonumber
\begin{aligned}
\nabla \bar{t}(\mu)&=2\Phi(-\frac{\mu}{\sigma_{ij}})-1-\sqrt{\frac{2}{\pi}}\frac{\mu}{\sigma_{ij}}\text{exp}({-\frac{\mu^2}{2\sigma_{ij}^2}})\\
&\quad+\sqrt{\frac{2}{\pi}}\frac{\mu}{\sigma_{ij}}\text{exp}(-\frac{\mu^2}{2\sigma_{ij}^2})+\text{sgn}(\mu)\\
&=2\Phi(-{\mu}/{\sigma_{ij}})-1+\text{sgn}(\mu)\\
&=\left\{\begin{array}{lc}
2\Phi(-\frac{\mu}{\sigma_{ij}}),&\text{ if }\mu>0,\\2\Phi(-\frac{\mu}{\sigma_{ij}})-2,&\text{ if }\mu<0.
\end{array}\right.
\end{aligned}
\ene
Thus, $\nabla \bar{t}(\mu)=\left\{\begin{aligned}>0,\text{ if }\mu>0\\<0,\text{ if }\mu<0\end{aligned}\right.$. Then, for any $\mu>0$,
\bee\nonumber
\begin{aligned}
\bar{t}(\mu)&<\lim_{\mu\ra\infty}\bar{t}(\mu)\\
&=\lim_{\mu\ra\infty} 2\mu\Phi(-\frac{\mu}{\sigma_{ij}})-\sigma_{ij}\sqrt{\frac{2}{\pi}}\text{exp}(-\frac{\mu^2}{2\sigma_{ij}^2})\\
&=0.
\end{aligned}
\ene
Similarly, $\bar{t}(\mu)<\lim_{\mu\ra-\infty}\bar{t}(\mu)=0$ for any $\mu<0$. Since $\bar{t}(0)=-\sigma_{ij}\sqrt{2/\pi}<0$, we obtain $\bar{t}(\mu)<0$ for all $\mu\in \bR$, and hence
$
t(\mu)<-|\mu|+\sigma_{ij}\sqrt{{2}/{\pi}}
$ where Fig. \ref{fig1:theo4} illustrates their gap. 
\begin{figure}[!t]
\centering
\includegraphics[width=\linewidth]{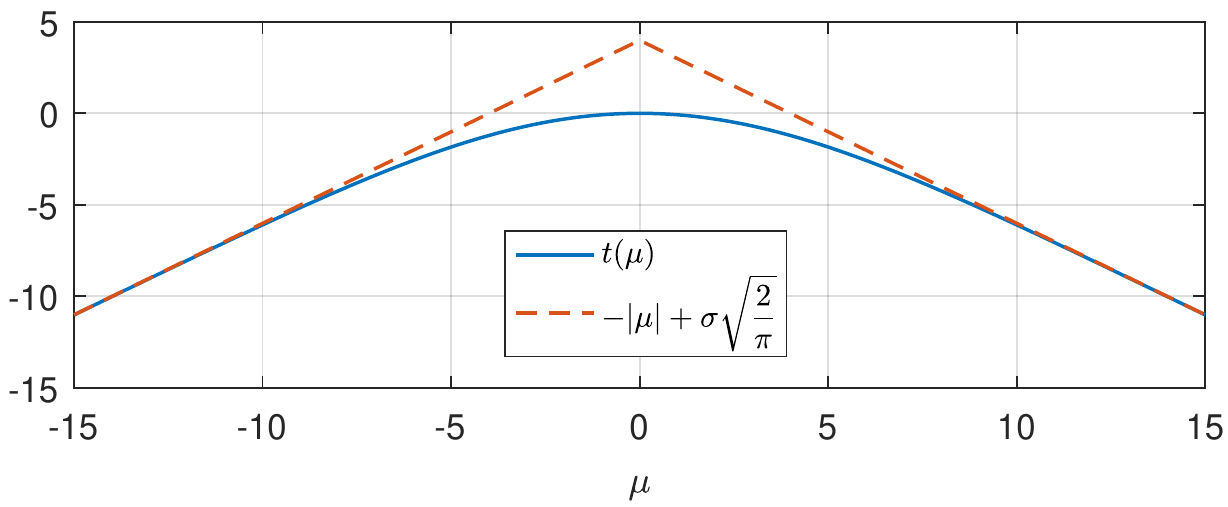}
\caption{The graph of $t(\mu)$ and $-|\mu|+\sigma\sqrt{\frac{2}{\pi}}$ on $[-3\sigma,3\sigma]$.}
\label{fig1:theo4}
\end{figure}
The above implies that \bee\nonumber
\begin{aligned}
g(\bar{\bm x}^\star)-f^\star&\leq\frac{\lambda}{2}\sum_{i,j}a_{ij}t(\mu_{ij})\\
&<\frac{\lambda}{2}\sum_{i,j}a_{ij}(-|\mu_{ij}|+\sigma_{ij}\sqrt{{2}/{\pi}})\\
&=\lambda(-h(\bm x)+\sum_{i,j}a_{ij}\sigma_{ij}\sqrt{{2}/{\pi}})\\
&\leq \lambda(-a_\text{min}^{(l)}v(\bar{\bm x}^\star)+\sigma_s\sqrt{{2}/{\pi}}).
\end{aligned}
\ene
where the last inequality follows from \eqref{eq2_theo3} and the definition of $\sigma_s$. Then, it follows that
\bee\label{eq2_theo5}
v(\bar{\bm x}^\star)\leq\frac{1}{\lambda a_\text{min}^{(l)}}[g(x^\star\bone)-g(\bar{\bm x}^\star)+\lambda\sigma_s\sqrt{{2}/{\pi}}].
\ene
Moreover, we have the following results
\bee\label{eq5_theo5}
\begin{aligned}
&g(x^\star\bone)-g(\bar{\bm x}^\star)\leq-\nabla g(x^\star\bone)^\mathsf{T}(\bar{\bm x}^\star-x^\star\bone)\\
&=-\nabla g(x^\star\bone)^\mathsf{T}[(\bar{\bm x}^\star-\frac{1}{n}\bone^\mathsf{T}\bar{\bm x}^\star\bone)+(\frac{1}{n}\bone^\mathsf{T}\bar{\bm x}^\star\bone-x^\star\bone)]\\
&=-\nabla g(x^\star\bone)^\mathsf{T}(\bar{\bm x}^\star-\frac{1}{n}\bone^\mathsf{T}\bar{\bm x}^\star\bone)\\
&\leq\|\nabla g(x^\star\bone)\|\|\bar{\bm x}^\star-\frac{1}{n}\bone^\mathsf{T}\bar{\bm x}^\star\bone\|\\
&\leq\|\nabla g(x^\star\bone)\|\|\bar{\bm x}^\star-\frac{1}{2}(\max_i{\bar{x}_i^\star}+\min_i{\bar{x}_i^\star})\bone\|\\
&\leq n\|\nabla g(x^\star\bone)\|_\infty\|\bar{\bm x}^\star-\frac{1}{2}(\max_i{\bar{x}_i^\star}+\min_i{\bar{x}_i^\star})\bone\|_\infty\\
&\leq \frac{ cn}{2}\max_{i,j}|\bar{x}_i^\star-\bar{x}_j^\star|=\frac{ cn}{2}v(\bar{\bm x}^\star)
\end{aligned}
\ene
where the second equality is from $\nabla g(x^\star\bone)^\mathsf{T}\bone=\nabla f(x^\star)=0$, and the third inequality follows from $\frac{1}{n}\bone^\mathsf{T}{\bm x}$ minimizes $\|\bm x-\alpha\bone\|$ w.r.t. $\alpha$ for all $\bm x$. Combining \eqref{eq2_theo5} and \eqref{eq5_theo5}, we obtain
\bee\nonumber
v(\bar{\bm x}^\star)\leq\frac{ cn}{2\lambda a_\text{min}^{(l)}}v(\bar{\bm x}^\star)+\frac{\sigma_s}{a_\text{min}^{(l)}}\sqrt{\frac{2}{\pi}}.
\ene
Since $2\lambda a_\text{min}^{(l)}> cn$, we have that
\bee\label{eq4_theo5}
v(\bar{\bm x}^\star)\leq\sqrt{\frac{2}{\pi}}\frac{2\lambda \sigma_s}{2\lambda a_\text{min}^{(l)}- cn}.
\ene

Next, we prove that there exists $x^\star\in \cX^\star$ such that
\bee\label{eq3_theo5}
\|\bar{\bm x}^\star-x^\star\bone\|_\infty\leq\max_{i,j}|\bar{x}_i^\star-\bar{x}_j^\star|= v(\bar{\bm x}^\star).
\ene
Clearly, it is sufficient to show that there exists $x^\star\in \cX^\star$ satisfying $x^\star\in[\min_i{\bar{x}_i^\star},\max_i{\bar{x}_i^\star}]$. Suppose that this is not true. Then, there is no loss of generality to assume that $\min_i{\bar{x}_i^\star}>x^\star$  for all $x^\star\in\cX^\star$. The first order necessary condition implies that
\bee\nonumber
\nabla f(x^\star)=\sum_{i=1}^n\nabla f_i(x^\star)=0.
\ene
Since $f_i(x)$ is convex and $x^\star<\bar{x}_i^\star$ for all $i$, it follows that $\nabla f_i(x^\star)\leq\nabla f_i(\bar{x}_i^\star)$ for all $i$. Letting $\bar{g}(\alpha)=g(\bar{\bm x}^\star+\alpha\bone)$, we obtain
\bee\label{eq6_theo5}
\nabla \bar{g}(0)=\bone^\mathsf{T}\nabla g(\bar{\bm x}^\star)=\sum_{i=1}^n\nabla f_i(\bar{x}_i^\star)\geq\sum_{i=1}^n\nabla f_i(x^\star)=0.
\ene
Actually, the inequality $\nabla f_i(x^\star)\leq\nabla f_i(\bar{x}_i^\star)$ must hold strictly for some $i\in\cV$.  Otherwise, we obtain that $\nabla f_i(x^\star)=\nabla f_i(\bar{x}_i^\star)$ for all $i$ which further implies that $\nabla f_i(x)=\nabla f_i(x^\star)$ for all $x\in[x^\star,\bar{x}_i^\star]$. Particularly, $\nabla f_i(\min_i\bar{x}_i^\star)=\nabla f_i(x^\star)$ for all $i$. Then, $\sum_{i=1}^n\nabla f_i(\min_i\bar{x}_i^\star)=0$, i.e., $\min_i\bar{x}_i^\star\in \cX^\star$. This contradicts the supposition that $x^\star<\min_i{\bar{x}_i^\star}$ for all $x^\star\in \cX^\star$.  


Hence, the inequality in \eqref{eq6_theo5} holds strictly, which implies
$
\nabla \bar{g}(0)>0,
$
and hence there exists an $\alpha<0$ such that $g(\bar{\bm x}^\star+\alpha\bone)<g(\bar{\bm x}^\star)$. Since $\bar{h}({\bm x})=\bar{h}({\bm x}+\alpha\bone)$ for all $\alpha\in R$, it follows that $\bar{f}_\lambda(\bar{\bm x}^\star+\alpha\bone)<\bar{f}_\lambda(\bar{\bm x}^\star)$, which contradicts the optimality of $\bar{\bm x}^\star$. That is, it is impossible to have $\min_i{\bar{x}_i^\star}>x^\star$  for all $x^\star\in\cX^\star$. Hence, \eqref{eq3_theo5} is established.

Combining \eqref{eq4_theo5} and \eqref{eq3_theo5} yields
\bee\nonumber
\|\bar{\bm x}^\star-x^\star\bone\|_\infty\leq v(\bar{\bm x}^\star) \leq \sqrt{\frac{2}{\pi}}\frac{2\lambda \sigma_s}{2\lambda a_\text{min}^{(l)}- cn},
\ene
which completes the proof. 
\end{proof}

Together with Lemma \ref{lemma4}, Theorem \ref{theo7} shows that consensus among agents may not be achieved in the presence of measurement noise. However each agent converges almost surely to a point that lies within a neighborhood of an optimal solution of problem \eqref{original}, the size of which is proportional to the noise level. Moreover, this optimal solution is encompassed by agents' final states.  

\section{A Distributed Algorithm over Randomly Activated Graphs}\label{sec_5}
This section studies the performance of \ref{protocol} over randomly activated graphs, which are defined as follows.
\begin{defi}[Randomly Activated Graphs]\label{def_rcg}
$\cG^k$ are randomly activated if for all $i,j\in\cV,i\neq j$, $\{a_{ij}^k\}$ is an i.i.d. Bernoulli process with $\bP\{a_{ij}^k=1\}=p_{ij}$, where $\bP(\cX)$ denotes the probability of an event $\cX$ and $0\leq p_{ij}\leq 1,\ \forall i,j\in\cV$. 
\end{defi}

We call $P=[p_{ij}]$ as the activation matrix of $\cG^k$, and the graph associated with $P$ is denoted as $\cG_P$, which is also the mean graph of $\cG^k$, i.e., 
\bee\label{meangraph}
\cG_P:=\bE(\cG^k).
\ene Randomly activated graphs can model many networks such as gossip social networks and random measurement losses in networks. They are different from another class of commonly used time-varying graphs that require the connectedness of the network in any finite time interval, see e.g. \cite{nedic2015distributed,olfati2004consensus}.

Under this scenario, \ref{protocol} is revised as
\begin{tcolorbox}[ams equation,size=small]\label{protocol_t}\tag{Algo. 3}
x_i^{k+1}=x_i^{k}+\lambda\rho^k\sum_{j\in \cN_i^k}\text{sgn}(x_j^k-x_i^k)-\rho^k\nabla f_i(x_i^k),
\end{tcolorbox}
\noeqref{protocol_t}
\noindent where the time-varying set of neighbors is given by $\cN_i^k=\{j\in\cV|(i,j)\in \cE^k\}$. For brevity, the weight of each edge $a_{ij}^k$ is now taken to be either zero or one.  

Similarly, \ref{protocol_t} is just the iteration of the stochastic subgradient method of the following optimization problem
\bee\label{timevarying}
\minimize_{\bm x\in\bR^n}\ \hat{f}_\lambda(\bm x):=g(\bm x)+\lambda \hat{h}(\bm x)
\ene
where $g(x)$ is given in \eqref{reformulate} and
\bee\label{hx_tv}
\hat{h}(\bm x)=\frac{1}{2}\sum_{i,j=1}^n p_{ij}|x_i-x_j|.
\ene
To exposit it, notice that $\bE(a_{ij}^k)=p_{ij}$, and thus a stochastic subgradient $\nabla_s \hat{h}(\bm x)=[\nabla_s \hat{h}(\bm x)_1,...,\nabla_s \hat{h}(\bm x)_n]^\mathsf{T}$ of $\hat{h}(\bm x)$ is given element-wise by
\bee\nonumber
\begin{aligned}
\nabla_s \hat{h}(\bm x)_i &=\sum_{j=1}^na_{ij}^k\text{sgn}(x_i-x_j)=\sum_{j\in \cN_i^k}\text{sgn}(x_i-x_j).
\end{aligned}
\ene

Since $\bE\{\nabla_s \hat{h}(\bm x)_i\}=\sum_{j}p_{ij}\text{sgn}(x_i-x_j)$, $\bE\{\nabla_s \hat{h}(\bm x)\}$ is a subgradient of $\hat{h}(\bm x)$. It follows from Lemma \ref{lemma4} that all agents almost surely converge to an optimal solution of problem \eqref{timevarying} under \ref{protocol_t}.  The following theorem summarizes the above analysis, and is the main result of this section.
\begin{theo}\label{theo6} Suppose that Assumptions \ref{assum1} and \ref{assum2}\ref{assum2a} hold, and that the multi-agent network $\cG_P$ is $l$-connected. Select
\bee\nonumber
\lambda>\frac{n c}{2p_{\min}^{(l)}},
\ene
where $\cG_P$ is given in \eqref{meangraph}, $p_{\min}^{(l)}$ denotes the sum of the $l$ smallest nonzero elements of $P$, and  $\{\rho^k\}$ satisfy
\bee\nonumber
\sum_{k=0}^\infty\rho^k=\infty,\quad\sum_{k=0}^\infty(\rho^k)^2<\infty. 
\ene
Let $\{\bm x^k\}$ be generated by \ref{protocol_t}. Then, it holds almost surely that $\lim_{k\ra\infty}\bm x^k=x^\star\bone$ for some $x^\star\in \cX^\star$.
\end{theo}
\begin{proof}
By Theorem \ref{theo1}, it follows that problem \eqref{timevarying} has the same set of optimal solutions and optimal value as problem \eqref{original}. The convergence proof of \ref{protocol_t} is very similar to that of  Lemma \ref{lemma4}.
\end{proof}

\section{Application to Distributed Quantile Regression}\label{sec_6}
In this section we apply our algorithms to solve the distributed quantile regression problem \cite{wang2017distributed}, which is widely used in statistics and econometrics \cite{hunter2000quantile,wang2017distributed}. Suppose we have observed $n$ sample points $(y_1,s_1),...,(y_n,s_n)$ where $y_i,s_i\in \bR$ for all $i\in\{1,...,n\}$ (we consider here only the scalar case for brevity). Our objective is to find the $\alpha$-th ($\alpha\in[0,1]$) linear quantile regression estimate $x_\alpha\in \bR$, which is an optimal solution to the following convex optimization problem \cite{hunter2000quantile}:
\bee\label{disquan1}
\minimize_{x\in \bR}\ f(x):=\sum_{i=1}^n f_i(x)=\sum_{i=1}^n Q_\alpha(y_i-xs_i)
\ene
where  $\alpha$-th quantile function $Q_\alpha(x)$ is defined by
\bee\label{quantilefunction}
Q_\alpha(x) = \left\{
	\begin{array}{lc}
		\alpha x, &\text{ if } x\geq 0,\\
		(\alpha-1)x, &\text{ if } x<0.
	\end{array}\right.
\ene
Hence, a subgradient of $f_i(x)$ is
\bee\nonumber
\nabla f_i(x)=\left\{\begin{aligned}&-\alpha s_i,&&\text{ if }y_i\geq xs_i,\\&(1-\alpha) s_i,&&\text{ if }y_i<xs_i.\end{aligned}\right.
\ene

Clearly, this problem satisfies Assumptions \ref{assum1} and \ref{assum2}\ref{assum2a} with $ c=\max_i\{\alpha s_i,(1-\alpha) s_i\}$, and thus we can apply our algorithms to solve it.

\subsection{The Effect of $\lambda$ and $\rho^k$.}\label{sec_sim_a}
We first illustrate that the lower bound of $\lambda$ in Theorem \ref{theo1} is tight in some cases. For simplicity, let $s_i=1$ for all $i$; then the problem \eqref{disquan1} is to find the $\alpha$-th quantile of $\{y_1,...,y_n\}$. Here we set $\alpha=0.5$ (the median) and let $\{y_1,...,y_n\}=\{4.45,14.99,24.28,26.21,44.24,58.61,68.78,75.49\}$.  Then, the median can be any value in $[26.21,44.24]$. Consider a ring-shaped 2-connected graph as in Fig. \ref{fig1b:sec_3}, with 8 nodes and unit edge weights. Then, it follows from Theorem \ref{theo1} that $\lambda$ should be strictly greater than $\underline{\lambda}=\frac{n c}{2\cdot2}=1$ to ensure \ref{protocol} to converge to the median of the sample points. We set $\lambda$ to be 0.95, 1.05, and 10, respectively, to examine their performance under \ref{protocol} and set the stepsize as $\rho^k={100}/{(k+10)}$. The trajectories of all agents are shown in Fig. \ref{fig1:exp1}.
\begin{figure}[!t]
\centering
\includegraphics[width=1\linewidth]{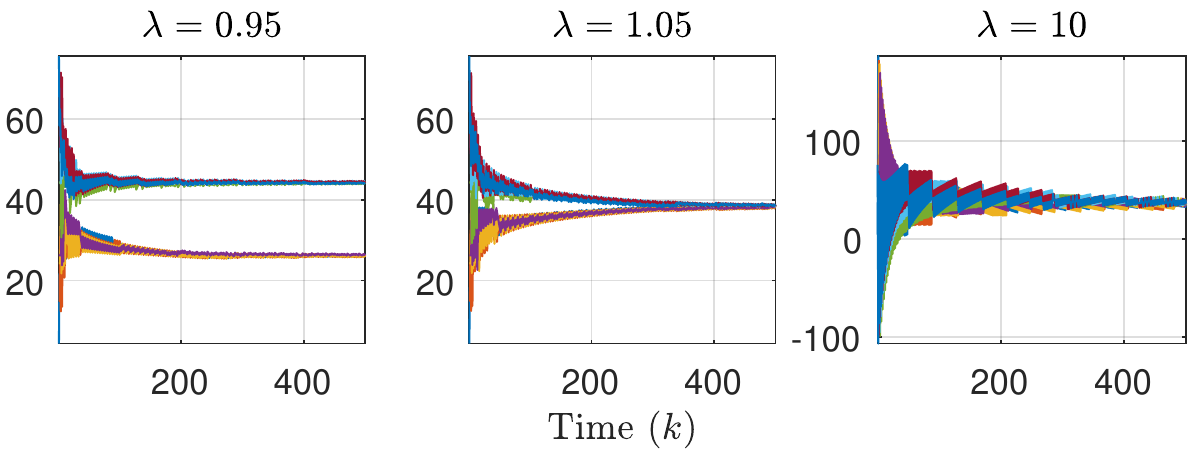}
\caption{Trajectories of all agents with different $\lambda$ under \ref{protocol}.}
\label{fig1:exp1}
\end{figure}

\begin{figure}[!t]
\centering
\includegraphics[width=1\linewidth]{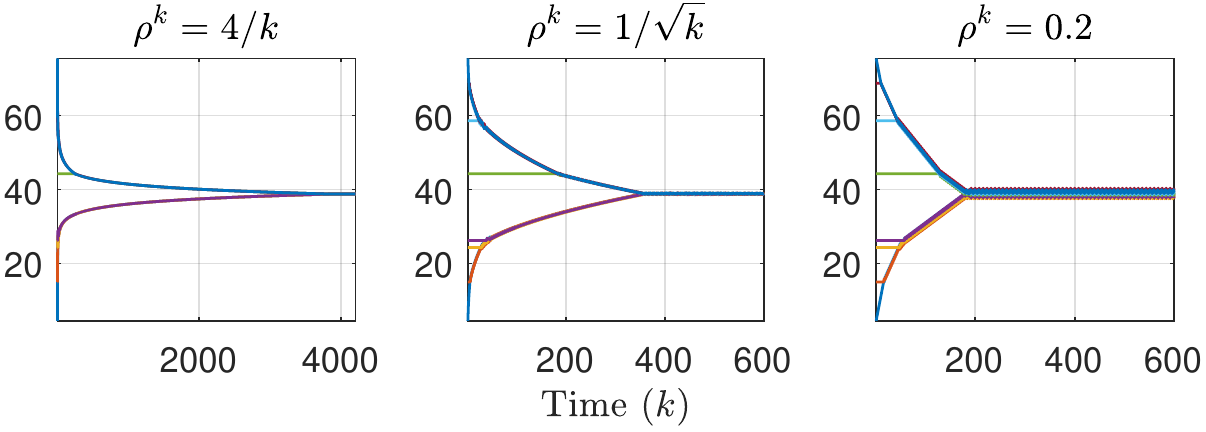}
\caption{Trajectories of all agents with different $\rho^k$ under \ref{protocol}.}
\label{fig2:exp1}
\end{figure}

As shown in Fig \ref{fig1:exp1}, consensus is not achieved even when $\lambda$ is slightly smaller than $\underline{\lambda}$ (the left subgraph), while the algorithm converges to the median when $\lambda$ is larger than $\underline{\lambda}$ (the middle and the right subgraphs). Besides, a larger value of $\lambda$ results in larger fluctuations in the transient stage. This suggests that it is better to choose a small $\lambda$ as long as it satisfies the condition of Theorem \ref{theo1}.

Fig. \ref{fig2:exp1} shows the trajectories under different stepsize rules for $\lambda=2$. The convergence with $\rho^k={4}/{k}$ is the slowest (the left subgraph), while it is faster for $\rho^k={1}/{\sqrt{k}}$ (the middle subgraph). Note that the algorithm under the constant stepsize approaches fastest to a neighborhood of an optimal solution.

\subsection{Noisy Measurements}\label{sec_sim_c}
We now study the effect of the measurement error described in Section \ref{sec_7} on the performance of our algorithms. Under the same settings as in Section \ref{sec_sim_a}, we have run two simulations. Both are expected to calculate the $0.4$-th quantile of $\{y_1,...,y_n\}$. We choose $\rho^k={40}/{(k+20)}$ and $\lambda=2$. The variance of the measurement error is $\sigma^2=9$ for all edges.  The trajectories of all agents in the two experiments are shown in the Fig. \ref{fig1:exp3}.
\begin{figure}[!t]
\centering
\includegraphics[width=\linewidth]{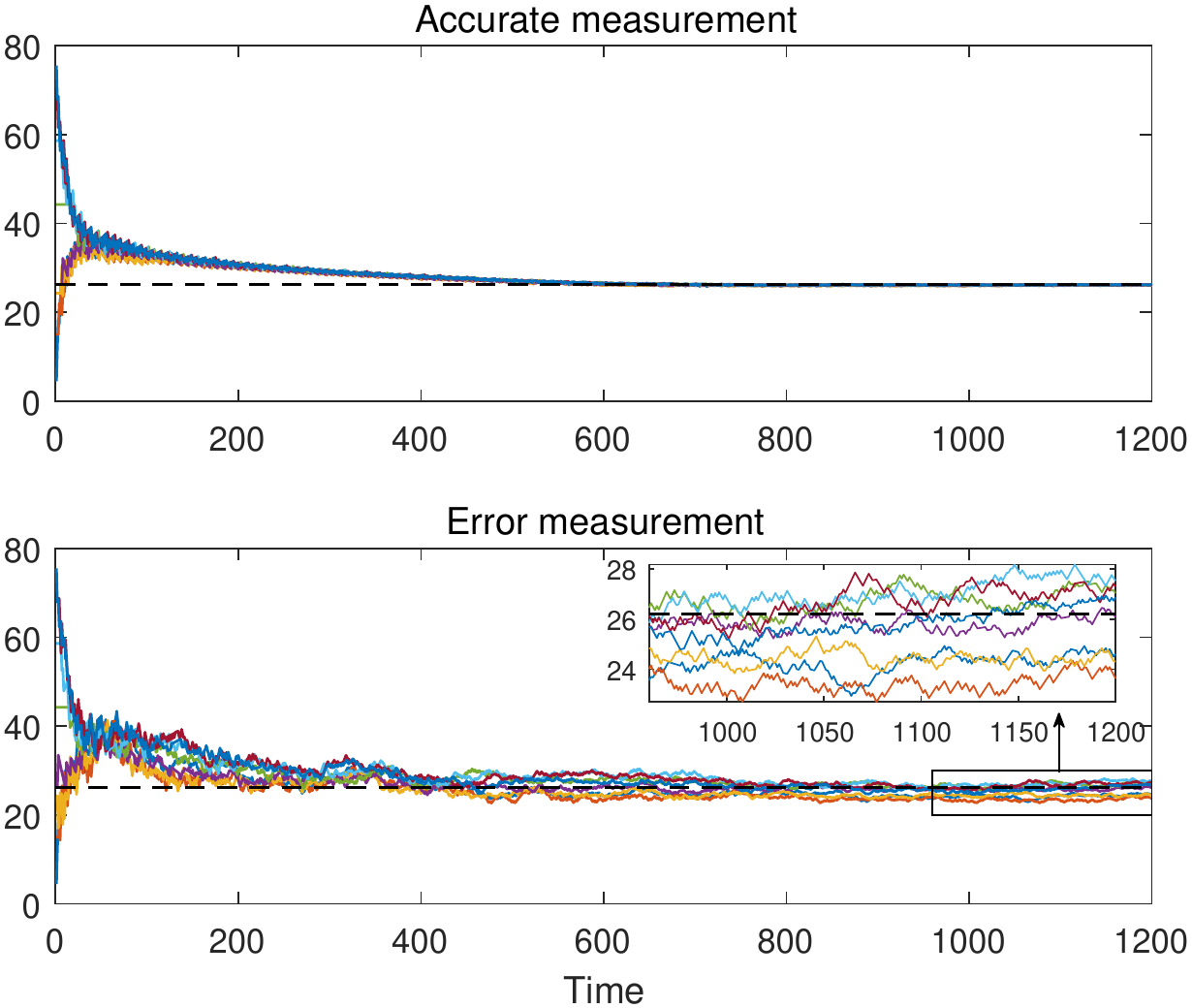}
\caption{The influence of the measurement error.}
\label{fig1:exp3}
\end{figure}
The black dash line represents one of the $0.4$-th quantile, i.e., 26.21. The top subgraph in Fig. \ref{fig1:exp3} is the result of \ref{protocol}, while the bottom subgraph is that of \ref{protocol_noisy}. It can be observed that all agents under \ref{protocol_noisy} only converge to a neighborhood of an optimal solution.

\subsection{Linear Quantile Regression}\label{sec_sim_b}
We have run two simulations over a static graph and randomly activated graphs, respectively. Both calculate the 0.1-th, 0.5-th and 0.9-th quantile regression estimates simultaneously by using 20 randomly generated sample points. The graph is ring-shaped with 20 nodes. The stepsizes are diminishing. We randomly choose some $a_e\in(0, 1]$ as the weight of edge $e$ of the static graph for all $e\in\{1,...,20\}$, which is also used as the activation probability of edge $e$ of the randomly activated graph.

\begin{figure}[!t]
\centering
\subfloat[]{\label{fig1a:exp2}{\includegraphics[width=0.49\linewidth]{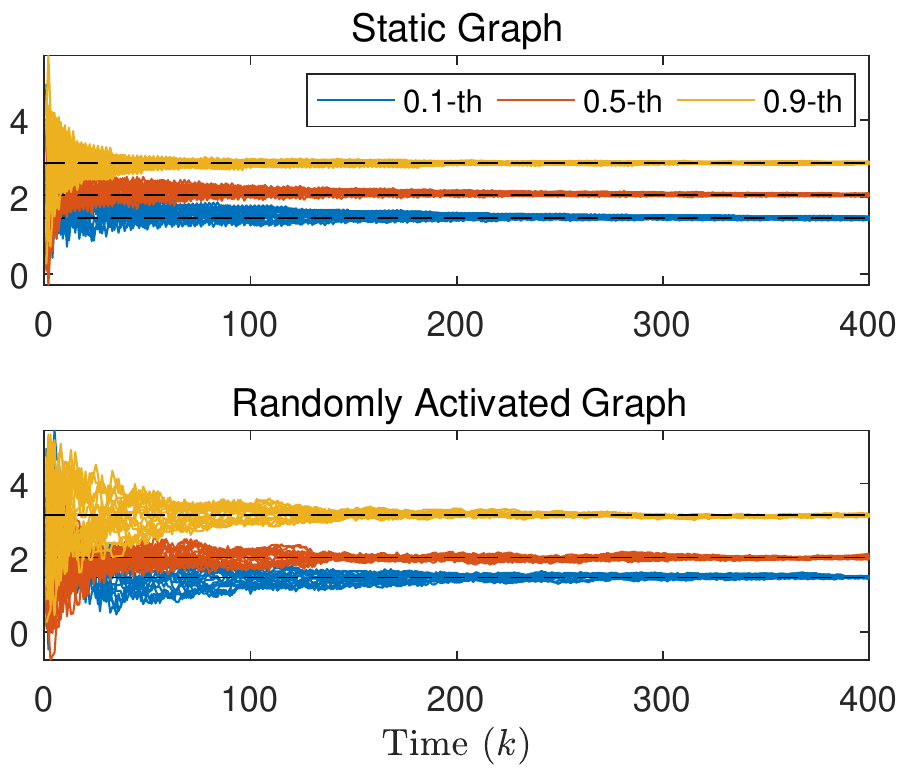}}}
\subfloat[]{\label{fig1b:exp2}\includegraphics[width=0.51\linewidth]{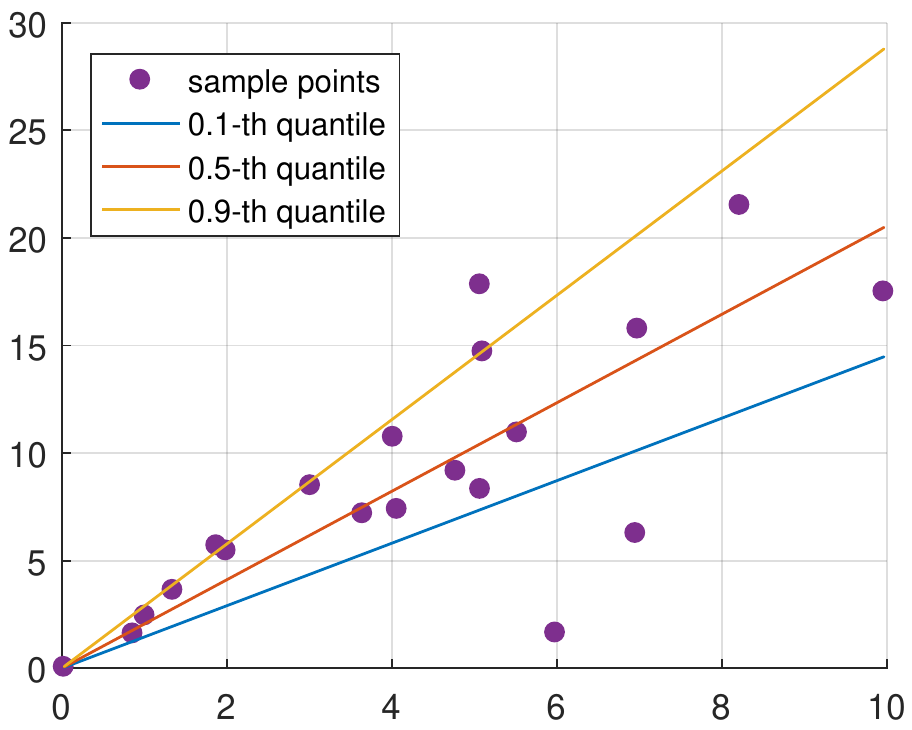}}
\caption{(a) Trajectories of all agents over a static graph and  randomly activated graphs. (b)  Three linear quantile regression estimates.}
\label{fig1:exp2}
\end{figure}

Fig. \ref{fig1a:exp2} illustrates the trajectories of the agents. All agents converge to the three quantile regression estimates (the black dash line) simultaneously. Besides, the randomly activated graph leads to larger fluctuations and a slower convergence rate. Fig. \ref{fig1b:exp2} plots our 20 sample points and the three linear estimates with $x_\alpha$ obtained in Fig. \ref{fig1a:exp2} for $\alpha=0.1,0.5$ and $0.9$, respectively, which shows that our algorithm converges to the correct points.

\section{Conclusions}
\label{sec_con}
In this paper, we have proposed a distributed optimization algorithm using as online information only the sign of relative state values in agent neighborhoods to solve the additive cost optimization problem in multi-agent networks. The network was allowed to be static or stochastically time-varying. For the former case, we have first provided a penalty method interpretation of our algorithm, and then studied its convergence  under  diminishing stepsizes as well as a constant stepsize. We have shown that the convergence rate varies from $O(1/\ln(k))$ to $O(\ln(k)/\sqrt{k})$, depending on the stepsize. For the latter case, we studied the  algorithm over the so-called randomly activated graphs, the convergence of which is given in the almost sure sense, and the case that the relative state information is noise corrupted. Finally, we have applied our algorithm to solve a quantile regression problem. All the theoretical results have been corroborated via simulations.

As shown in this work, using only the sign of the relative state information one is still able to solve the distributed optimization problem \eqref{original}. It is interesting to study the tradeoff between the convergence performance and the amount of information used for a distributed algorithm, which we leave as future work.

\appendices
\section{Proof of part (b) of Theorem \ref{theo1}}\label{appendix_a}
We first introduce additional basics of graph theory, which can be found in \cite{bullo2017lectures}.

Let $\cG=(\cV,\cE)$ be a graph with an adjacency matrix $A$. We number each edge of $\cG$ with a unique $e\in \cI$ and assign an arbitrary direction to each edge, where $\cI=\{1,...,m\}$ is called the edge number set of $\cG$ and $m$ is the number of edges. We say that node $i$ is the source node of edge $e$ if $e$ leaves $i$, and is the sink node if $e$ enters $i$. The incidence matrix $B\in \bR^{n\times m}$ of $\cG$ is defined by
\[B_{ie}=\left\{\begin{aligned}1, &\text{ if node $i$ is the source node of $e$},\\-1, &\text{ if node $i$ is the sink node of $e$},\\ 0, &\text{ otherwise}.\end{aligned}\right.
\]
For any $\bm x=[x_1,...,x_n]^\mathsf{T}$, we have that
\bee\label{eq1_sec2}
\bm b_e^\mathsf{T}\bm x=x_i-x_j
\ene
where $\bm b_e,e\in \cI$ is the $e$-th column of $B$, and $i$ and $j$ are the source and the sink nodes of edge $e$, respectively. Throughout this section, we use $i,j$ to denote nodes, and $e,u,v$ to denote edge numbers.

A connected graph is a tree if it becomes unconnected when any single edge is removed. A spanning tree $\cT$ of a connected graph $\cG$ is the tree with the same nodes as $\cG$ and a subset of the edges of $\cG$.

A subgraph $\cG_s=(\cV_s,\cE_s)$ of $\cG$ is a graph with $\cV_s\subseteq\cV$ and $\cE_s\subseteq\cE$. The subgraph of $\cG$ induced by $\cV_I\subseteq\cV$ is the graph $\cG_I=(\cV_I,\cE_I)$ where $\cE_I$ includes all edges of $\cE$ that connect two nodes of $\cV_I$. The subgraph $\cG_{cI}$ of $\cG$ induced connectedly by $\cV_I\subseteq\cV$ is $\cG_I$ with extra edges that make $\cG_{cI}$ connected.

A cut of a graph is a partition of its nodes into two non-empty and disjoint sets.

Finally, we define the set-valued function
\[
\text{SGN}(x)=\left\{
	\begin{array}{lc}
	\{1\},& \text{if }x>0,\\ \relax
	[-1,1],&\text{if }x=0,\\
	\{-1\},& \text{if }x<0.
	\end{array}\right.
\]
It is obvious that $\text{SGN}(x)$ is the subdifferential of $|x|$. With a slight abuse of notation, we use $\text{SGN}(\bm x)$ to represent the set-valued vector $[\text{SGN}(x_1),...,\text{SGN}(x_n)]^\mathsf{T}$.

To establish the proof of part (b) of Theorem \ref{theo1}, we need two lemmas on incidence matrices.

\begin{lemma}\label{lemma1}
Let $\cG_1$ be a graph with nodes $\cV_1=\{1,...,n_1\}$ and the edge number set $\cI_1=\{1,...,m_1\}$, and let $\cG_2$ be a graph with nodes $\cV_2=\{n_1+1,...,n_1+n_2\}$ and the edge number set $\cI_2=\{{m_1+1},...,{m_1+m_2}\}$. Denote by $B_1$ and $B_2$ the incidence matrices of $\cG_1$ and $\cG_2$, respectively. Let $B=\diag(B_1,B_2)$, and $n=n_1+n_2$. 

Assume that there is a new edge $e$ which connects some $p\in\cV_1$ and $q\in\cV_2$, and let $\bm b_e\in\bR^{n}$ be a vector with the $p$-th element 1, the $q$-th element -1 and other elements 0. Then for any $\alpha\in \bR$, it follows that
\bee\nonumber
\min_{\bm x \in \bR^{n}}\|\alpha \bm b_e-B\bm x\|_\infty\geq\frac{2|\alpha|}{n}.
\ene
\end{lemma}
\begin{proof}
Since edge $e$ joins two nodes which are in different graphs, there is no loss of generality to let the source node and the sink node be in $\cG_1$ and $\cG_2$, respectively. Then, we obtain
\[\|\alpha \bm b_e-B\bm x\|_\infty=\|\begin{bmatrix}B_1\bm x_1+\alpha \bm e_1\\B_2\bm x_2-\alpha \bm e_2\end{bmatrix}\|_\infty\]
where $\bm x=[\bm x_1^\mathsf{T},\bm x_2^\mathsf{T}]^\mathsf{T}\in \bR^{n}$, and both $\bm e_1$ and $ \bm e_2$ are vectors with one element 1 and other elements 0. By applying the inequality that $\|\bm x\|_\infty\geq\frac{1}{n}\|\bm x\|_1$ for all $\bm x\in\bR^n$, it follows that
\begin{equation*}
\begin{aligned}
&\|\begin{bmatrix}B_1 \bm x_1+\alpha \bm e_1\\B_2 \bm x_2-\alpha \bm e_2\end{bmatrix}\|_\infty\geq\|B_1 \bm x_1+\alpha \bm e_1\|_\infty\\
&\geq\frac{1}{n_1}\|B_1 \bm x_1+\alpha \bm e_1\|_1\geq\frac{1}{n_1}|\bone^\mathsf{T}(B_1 \bm x_1+\alpha \bm e_1)|.
\end{aligned}
\end{equation*}
Using the fact that $\bone^\mathsf{T}B_1=0$, we have for all $\bm x\in\bR^n$ that
\begin{equation*}
\begin{aligned}
\|\alpha \bm b_e-B\bm x\|_\infty&\geq\frac{|\bone^\mathsf{T}(B_1\bm x+\alpha \bm e_1)|}{n_1}=\frac{|\alpha\bone^\mathsf{T} \bm e_1|}{n_1}=\frac{|\alpha|}{n_1}.
\end{aligned}
\end{equation*}

Similarly, we obtain that $\|\alpha \bm b_e-B\bm x\|_\infty\geq {|\alpha|}/{n_2}$. Hence,
$$\min_{\bm x \in \bR^{n}}\|\alpha \bm b_e-B\bm x\|_\infty\geq\frac{|\alpha|}{\min\{n_1,n_2\}}\geq\frac{2|\alpha|}{n}.$$
\end{proof}

The following corollary directly follows from Lemma \ref{lemma1}.
\begin{coro}
Let $B\in\bR^{n\times (n-1)}$ be the incidence matrix of a tree. Then, for all $\bm x\in \bR^{n-1}$, the following inequality holds
\bee\nonumber
\|B\bm x\|_\infty\geq \frac{2}{n}\|\bm x\|_\infty.
\ene
\end{coro}
\begin{proof}
For any $e\in\{1,...,n-1\}$, the tree becomes two disjoint subtrees when the $e$-th edge is removed. Let $x_e$ denote the $e$-th element of $\bm x$, and let $\bm x_{-e}$ and $B_{-e}$ denote $\bm x$ with the $e$-th element removed and $B$ with the $e$-th column removed, respectively. Then, it follows from Lemma \ref{lemma1} that
\[
\|B\bm x\|_\infty=\|x_e\bm b_e+B_{-e}\bm x_{-e}\|_\infty\geq\frac{2|x_e|}{n}.
\]
Since $e$ is arbitrary, the result holds immediately.
\end{proof}

\begin{lemma}\label{lemma2}
Let $\cG$ be a graph with the edge number set $\cI$. A cut separates the nodes of $\cG$ in two subsets $\cV_1$ and $\cV_2$ which are joined by $l$ edges with the edge number set $\cI_c\subseteq \cI$. Let $\cG_1$ and $\cG_2$ be two graphs induced connectedly by $\cV_1$ and $\cV_2$, respectively. The edge number set and incidence matrix of $\cG_i$ are denoted respectively by $\cI_i$ and $B_i$, where $i\in\{1,2\}$. Let $\bar{B}=\diag(B_1,B_2)$. Then for any $e,u\in \cI_c$, we have 
\bee\label{eq_lemma2}
\bm b_{e}=\gamma_{eu} \bm b_{u}+\bar{B}\bm y_{eu},\ u\in\{1,...,k\}
\ene
where $\bm y_{eu}$ is a vector with elements $1, 0$ or $-1$, and
\beq\label{gamma}
&&\hspace{-0.7cm}\gamma_{eu}=\\
&&\hspace{-0.7cm}\left\{\begin{aligned}1, &~\text{if the source nodes of $e$ and $u$ are in the same subset},\\-1,&\text{ otherwise}.\end{aligned}\right.\nonumber
\enq

\end{lemma}
\begin{proof}
If $e=u$, the result holds by just setting $y_{eu}=0$. If $e\neq u$, we define the source node and sink node of $\bm b_e$ to be $p_e$ and $q_e$, and that the source node and sink node of $\bm b_u$ to be $p_u$ and $q_u$, respectively. We first assume $p_e$ and $p_u$ are in the same subset, say $\cG_1$, and thus $\gamma_{eu}=1$. Hence we can find a path in $\cG_1$ from $p_e$ to $p_u$ as $\cG_1$ is connected. Similarly, we can find a path in $\cG_2$ from $q_u$ to $q_e$. Therefore, we have a path from $p_e$ to $q_e$ through edge $u$ rather than edge $e$. The edge number set of edges in the path is denoted by $\cI_p\subseteq \cI_1\cup \cI_2\cup \{u\}$. 

Now we define $s_{ev},v\in \cI_p$ as
\bee\label{eq3_lemma2}
\begin{aligned}
&s_{ev}=\left\{\begin{aligned}1,&\text{ if the source node of edge $v$ is closer to $p_e$},\\-1,&\text{ if the sink node of edge $v$ is closer to $p_e$}.\\\end{aligned}\right.
\end{aligned}
\ene
Note that $s_{eu}=\gamma_{eu}$. Then, \eqref{eq_lemma2} is obtained since $\bm b_e=\sum_{v\in \cI_p}s_{ev}\bm b_v=\gamma_{eu}\bm b_u+\sum_{v\in \cI_p-\{u\}}s_{ev}\bm b_{ev}=\gamma_{eu} \bm b_{u}+\bar{B}\bm y_{eu}$, where the last equality holds because $\bar{B}$ includes $\bm b_v$ as one of its columns for all $v\in \cI_p-\{u\}$. 

If $p_e$ and $p_u$ are in different subsets where $\gamma_{eu}=-1$, we obtain \eqref{eq_lemma2} by similar arguments.
\end{proof}

\begin{proof}[Proof of part (b) of Theorem \ref{theo1}]

Notice that the penalty function $h(\bm x)$ can be represented as
\bee\label{hx2}
\begin{aligned}
h(\bm x)&=\sum_{e=1}^m a_e |\bm b_e^\mathsf{T}\bm x|.
\end{aligned}
\ene
where $a_e$ is the weight of edge $e$. The subdifferential of $h(\bm x)$ is then given by
\bee\label{subdiff}
\partial h(\bm x)=\sum_{e=1}^ma_e \text{SGN}(\bm b_e^\mathsf{T}\bm x)\bm b_e=BA_e\text{SGN}(B^\mathsf{T}\bm x)
\ene
where $A_e=\diag\{a_1,...,a_m\}$. This implies that the subdifferential of $\tilde{f}_\lambda(\bm x)$ is
\bee\label{gradf1}
\begin{aligned}
\partial \tilde{f}_\lambda(\bm x) =\lambda BA_e\cdot\text{SGN}(B^\mathsf{T}\bm x)+\partial g(\bm x).
\end{aligned}
\ene

Let $B$, $\cV=\{1,...,n\}$ and $\cI=\{1,...,m\}$ be the incidence matrix, the node set, and the edge number set of the graph $\cG$, respectively, and let $x_i$ be the $i$-th element of $\bm x$. We define $\cV_\text{max}=\argmax_{i\in\cV} x_i$ and $\cV_{\text{r}}=\cV-\cV_\text{max}$. Since $\bm x\neq \alpha\bone$, then $\cV_{\text{r}}$ is not empty, and has $l_0\geq l$ edges connected to $\cV_\text{max}$. Denote the edge number set and the set of weights of these $l_0$ edges by $\cI_\text{c}\subseteq \cI$ and $\cA_{\text{c}}=\{a_e|e\in \cI_\text{c}\}$, respectively. Note that each of these $l_0$ edges connects two nodes with different values, which implies that $\bm b_e^\mathsf{T}\bm x\neq0$ for all $e\in \cI_\text{c}$. We can appropriately choose the orientation of each edge $e$ for all $e\in \cI_\text{c}$ such that $\bm b_e^\mathsf{T}\bm x>0$. It then follows that $\text{SGN}(B_{\text{c}}^\mathsf{T}\bm x)=\bone$, where $B_{\text{c}}=[\bm b_e],e\in \cI_\text{c}$. 

Let $\cI_{\text{r}}=\cI-\cI_c,\cA_{\text{r}}=\{a_e,e\in \cI_{\text{r}}\}$, and $B_{\text{r}}=[\bm b_e],e\in \cI_{\text{r}}$. By properly shifting the orders of columns of $B,A_e,\bm x$ and $\partial g(\bm x)$, we obtain
\bee\label{gradf2}
\begin{aligned}
&\partial\tilde{f}_\lambda(\bm x)\\
&=\lambda BA_e\text{SGN}(B^\mathsf{T}\bm x)+\partial g(\bm x)\\
&=\lambda [B_{\text{c}},B_{\text{r}}]\begin{bmatrix}A_\text{c}&0\\0&A_\text{r}\end{bmatrix}\begin{bmatrix}\text{SGN}(B_\text{c}^\mathsf{T}\bm x)\\\text{SGN}(B_\text{r}^\mathsf{T}\bm x)\end{bmatrix}+\partial g(\bm x)\\
&=\lambda B_{\text{c}}A_{\text{\text{c}}}\text{SGN}(B_\text{c}^\mathsf{T}\bm x)+\lambda B_{\text{r}}A_{\text{r}}\text{SGN}(B_\text{r}^\mathsf{T}\bm x)+\partial g(\bm x)\\
&=\lambda B_{\text{c}}A_{\text{\text{c}}}\bone+\lambda B_{\text{r}}A_{\text{r}}\cY+\partial g(\bm x)\\
&=\lambda\sum_{e\in \cI_\text{c}}\bm b_{e}a_e+\lambda B_{\text{r}}A_{\text{r}}\cY+\partial g(\bm x)
\end{aligned}
\ene
where $A_{\text{c}} =\text{diag}(\cA_{\text{c}}), A_{\text{r}} =\text{diag}(\cA_{\text{r}})$, and $\cY\subseteq[-1,1]^{m-l_0}$.

Consider two subgraphs $\cG_1$ and $\cG_2$ of the graph $\cG$ induced connectedly by $\cV_\text{max}$ and $\cV_\text{r}$, respectively. Let the incidence matrices of $\cG_1$ and $\cG_2$ be $B_1$ and $B_2$, respectively, and let $\bar{B}=\diag(B_1,B_2)$. From Lemma \ref{lemma2} we know that for any $u\in \cI_\text{c}$, we have $\bm b_e=\gamma_{ue}\bm b_u+\bar{B}\bm z_{ue}$ for all $e\in \cI_\text{c}$, where $\gamma_{ue}$ is given by \eqref{gamma}, and $z_{ue}$ is a vector with elements $1, 0$ or $-1$. Since $\text{SGN}(B_\text{c}^\mathsf{T}\bm x)=\bone$, all edges have their source nodes in the same subset, and hence all $\gamma_{ue},e\in \cI_\text{c}$ are equal. Thus we can let $\gamma_{ue}=\gamma_u,\forall e\in \cI_\text{c}$. Substituting this into \eqref{gradf2} yields
\bee\label{eq1_appendix}
\begin{aligned}
&\partial \tilde{f}_\lambda(\bm x)\\
&=\lambda\sum_{e\in \cI_\text{c}}(\gamma_u \bm b_u+\bar{B}\bm z_{ue})a_{e}+\lambda B_{\text{r}}A_{\text{r}}\cY+\partial g(\bm x)\\
&=\lambda\gamma_u \bm b_u\sum_{e\in \cI_\text{c}}a_{e}+\bar{B}\lambda\sum_{e\in \cI_\text{c}}\bm z_{ue}a_{e}+\lambda B_\text{r}A_{\text{r}}\cY+\partial g(\bm x)\\
&= \lambda\gamma_u \bm b_u\sum_{e\in \cI_\text{c}}a_{e}+\bar{B}\cS+\partial g(\bm x)
\end{aligned}
\ene
where $\cS$ is a subset of $\bR^{m-l_0}$ and $u\in \cI_c$. The last equality holds because $\bar{B}$ includes all columns of $B_\text{r}$ by its definition, and hence $\text{range}(B_\text{r})\subseteq\text{range}(\bar{B})$. Note that $\gamma_u=1$ or $-1$. Thus, it follows from \eqref{eq1_appendix} that for any $\nabla \tilde{f}_\lambda(\bm x)\in \partial \tilde{f}_\lambda(\bm x)$, we can find some $\nabla g(\bm x)\in \partial g(\bm x)$ and $\bm s\in\cS$ such that
\bee\nonumber
\begin{aligned}
\nabla \tilde{f}_\lambda(\bm x)=\lambda\gamma_u \bm b_u\sum_{e\in \cI_\text{c}}a_{e}+\bar{B}\bm s+\nabla g(\bm x).
\end{aligned}
\ene
Using also Assumption \ref{assum2}\ref{assum2a}, we have
\bee\nonumber
\begin{aligned}
\|\nabla \tilde{f}_\lambda(\bm x)\|_{\infty}&\geq \|\lambda\gamma_u \bm b_u\sum_{e\in \cI_\text{c}}a_{e}+\bar{B}\bm s\|_{\infty}-\|\nabla g(\bm x)\|_{\infty}\\ 
&\geq \min_{\bm t\in\bR^{m-l_0}}\|\lambda\gamma_u \bm b_u\sum_{e\in \cI_\text{c}}a_{e}+\bar{B}\bm t\|_{\infty}- c\\
&= \min_{\bm t\in\bR^{m-l_0}}\|\lambda \bm b_u\sum_{e\in \cI_\text{c}}a_{e}+\bar{B}\bm t\|_{\infty}- c.
\end{aligned}
\ene

By Lemma \ref{lemma1} and the $l$-connected property, it follows that for all $\bm x\neq\alpha\bone$,
\[
\begin{aligned}
\|\nabla \tilde{f}_\lambda(\bm x)\|_{\infty}&\geq \min_{\bm t\in\bR^{m-l_0}}\|\lambda \bm b_u\sum_{e\in \cI_\text{c}}a_{e}+\bar{B}\bm t\|_{\infty}- c\\
&\geq\frac{2|\lambda|}{n}\sum_{e\in \cI_\text{c}}a_{e}- c\geq \frac{2\lambda a^{(l)}_\text{min}}{n}- c,
\end{aligned}
\]
which completes the proof.
\end{proof}
\section{Proof of Theorem \ref{theo4}}\label{appendix_b}

We first show that $\tilde{d}(\rho)<\infty$. Since $\tilde{f}_\lambda(x)$ is convex, $\tilde{\cX}(\rho)$ is convex and $\cX^\star\subseteq\tilde{\cX}(\rho)$ for any $\rho>0$. One can verify that $\tilde{\cX}(\rho)-\cX^\star$ is bounded. If $\tilde{\cX}(\rho)-\cX^\star$ is empty, then $\tilde{d}(\rho)=0$, otherwise $0\leq\tilde{d}(\rho)=\max_{x\in\tilde{\cX}(\rho)}\min_{x^\star\in \cX^\star}|x-x^\star|=\max_{x\in\tilde{\cX}(\rho)-\cX^\star}\min_{x^\star\in \cX^\star}|x-x^\star|<\infty$.

Then, we claim the following. 

{\em Claim 1:}~If $\|\bm x^{k}-x^\star\bone\|>c_\rho$ for all $x^\star\in \cX^\star$, then $\tilde{f}_\lambda(\bm x^k)-f^\star>{\rho c_a^2}/{2}$.  

Recall from \eqref{eq6_theo3} that
\[
\begin{aligned}
&\tilde{f}_\lambda(\bm x^k)-f^\star\geq f(\bar{x}^k)-f^\star+(\lambda a_\text{min}^{(l)}-\frac{1}{2} cn)v(\bm x^k), \forall k.
\end{aligned}
\]
This implies that if either $f(\bar{x}^k)-f^\star>{\rho c_a^2}/{2}$ or $v(\bm x^k)>\frac{\rho c_a^2}{2\lambda a_\text{min}^{(l)}- cn}$, then $\tilde{f}_\lambda(\bm x^k)-f^\star>{\rho c_a^2}/{2}$.  Let $$c_\rho:=2\sqrt{n}\max\{\tilde{d}(\rho),\frac{\rho c_a^2}{2\lambda a_\text{min}^{(l)}- cn}\}.$$ Since
\bee\nonumber
\begin{aligned}
c_\rho&<\|\bm x^{k}-x^\star\bone\|\leq \|\bm x^{k}-\bar{x}^k\bone\|+\|\bar{x}^k\bone-x^\star\bone\|\\
&\leq\sqrt{n}v(\bm x^k)+\sqrt{n}|\bar{x}^k-x^\star|
\end{aligned}
\ene
we obtain that $v(\bm x^k)>c_\rho/(2\sqrt{n})\ge \frac{\rho c_a^2}{2\lambda a_\text{min}^{(l)}- cn}$ or $|\bar{x}^k-x^\star|>c_\rho/(2\sqrt{n})\ge \tilde{d}(\rho)$. For the former case we have $\tilde{f}_\lambda(\bm x^k)-f^\star>{\rho c_a^2}/{2}$. For the latter case, $\bar{x}^k\notin\tilde{\cX}(\rho)$, which by the definition of $\tilde{\cX}(\rho)$ implies $\tilde{f}_\lambda(\bm x^k)-f^\star>{\rho c_a^2}/{2}$.

{\em Claim 2:} There is $x_0^\star\in \cX^\star$ such that $\liminf_{k\rightarrow \infty} \|\bm x^{k}-x_0^\star\bone\|\leq c_\rho$. 

Otherwise, there exists $k>0$ such that 
$$\|\bm x^{k}-x^\star\bone\|> c_\rho, \forall x^\star\in \cX^\star, \forall  k>k.$$ 

By Claim 1,  there exists some $\epsilon>0$ such that $\tilde{f}_\lambda(\bm x^k)-f^\star>{\rho c_a^2}/{2}+\epsilon$ for all $k>k$. Together with \eqref{eq_sg}, it yields that
\bee\label{eq3_theo4}
\begin{aligned}
\|\bm x^{k+1}-x^\star\bone\|^2&\leq\|\bm x^{k}-x^\star\bone\|^2-2\rho(\tilde{f}_\lambda(\bm x^k)-f^\star)+\rho^2c_a^2\\
&\leq\|\bm x^{k}-x^\star\bone\|^2-2\rho(\frac{\rho c_a^2}{2}+\epsilon)+\rho^2c_a^2\\
&=\|\bm x^{k}-x^\star\bone\|^2-2\rho \epsilon.
\end{aligned}
\ene
Summing this relation implies that for all $k>k$, 
\bee\nonumber
\begin{aligned}
&\|\bm x^{k+1}-x^\star\bone\|^2\leq\|\bm x^{k}-x^\star\bone\|^2-2(k+1-k)\rho \epsilon,
\end{aligned}
\ene
which clearly cannot hold for a sufficiently large $k$. Thus, we have verified Claim 2. 

{\em Claim 3}:  There is $x^\star\in \cX^\star$ such that $\limsup_{k\rightarrow \infty} \|\bm x^{k}-x^\star\bone\|\leq c_\rho +\rho c_a$. 

Otherwise, for {\em any} $x^\star\in\cX^\star$, there must exist a subsequence $\{\bm x^k\}_{k\in\cK}$ (which depends on $x^\star$) such that  for all $k\in\cK$, 
\bee\label{eq6_theo4}
\|\bm x^{k}-x^\star\bone\|>c_\rho+\rho c_a.
\ene
Moreover, it follows from \eqref{subdiff}  that 
\beq
&&\|\bm x^{k+1}-x^\star\bone\| \nonumber \\
&&=\|\bm x^{k}-x^\star\bone-\rho\lambda BA_e\text{sgn}(B^\mathsf{T}\bm x^k)-\rho\nabla g(\bm x^k)\| \nonumber \\
&&\leq\|\bm x^{k}-x^\star\bone\|+\lambda\rho\| BA_e\text{sgn}(B^\mathsf{T}\bm x^k)\|+\rho\|\nabla g(\bm x^k)\| \nonumber \\
&&\leq\|\bm x^{k}-x^\star\bone\|+\rho\sqrt{n}(\lambda \|A\|_\infty +  c) \nonumber \\
&&=\|\bm x^{k}-x^\star\bone\|+\rho c_a, \forall k
\label{eq2_theo4}
\enq
where the second inequality follows from that
\bee\nonumber
\begin{aligned}
\|BA_e\text{sgn}(B^\mathsf{T}\bm x^k)\|&\leq\sqrt{n}\| BA_e\text{sgn}(B^\mathsf{T}\bm x^k)\|_\infty\\
&\leq\sqrt{n}\|BA_e\|_\infty\|\text{sgn}(B^\mathsf{T}\bm x^k)\|_\infty\\
&\leq\sqrt{n}\max_{i}\sum_{j=1}^na_{ij}=\sqrt{n}\|A\|_\infty.
\end{aligned}
\ene
Thus, we obtain that for all $k\in\cK$,  
\bee\label{eq4_theo4}
\|\bm x^{k-1}-x^\star\bone\| \geq\|\bm x^{k}-x^\star\bone\|-\rho c_a>c_\rho.
\ene

By Claim 2, there must exist some $k_1\in\cK$ and $k_1>k$ such that
\bee\nonumber
\begin{aligned}
&\|\bm x^{k_1-1}-x_0^\star\bone\|\leq c_\rho+\rho c_a.
\end{aligned}
\ene
Together with \eqref{eq4_theo4}, it implies that
\bee\label{crho}
c_\rho<\|\bm x^{k_1-1}-x_0^\star\bone\|\leq c_\rho+\rho c_a.
\ene
Hence, it follows from Claim 1 that $\tilde{f}_\lambda(\bm x^{k_1-1})-f^\star>{\rho c_a^2}/{2}$, which together with \eqref{eq3_theo4} and \eqref{crho} yields that
\bee\label{eq_end}
\|\bm x^{k_1}-x_0^\star\bone\|\leq\|\bm x^{k_1-1}-x_0^\star\bone\|\leq c_\rho+\rho c_a.
\ene
Set $x^\star=x_0^\star$ in \eqref{eq6_theo4}, we have $\|\bm x^{k_1}-x_0^\star\bone\|>c_\rho+\rho c_a.$ This contradicts \eqref{eq_end}, and hence verifies Claim 3. 

In view of \eqref{defidis}, the proof is completed. 
\qed




\bibliographystyle{IEEEtran}
\bibliography{mybibf}         
\begin{IEEEbiography}
[{\includegraphics[width=1in,height=1.25in,clip,keepaspectratio]{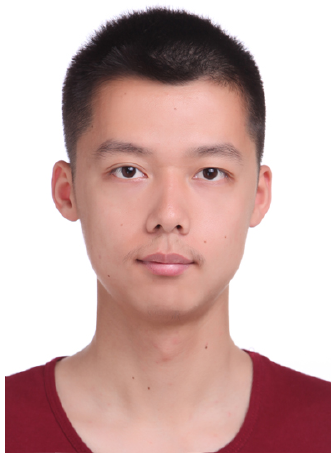}}]
{Jiaqi Zhang} received the B.S. degree from the School of Electronic and Information Engineering, Beijing Jiaotong University, Beijing, China, in 2016. He is currently pursuing the Ph.D. degree at the Department of Automation, Tsinghua University, Beijing, China. His research interests include networked control systems, distributed optimization and their applications.
\end{IEEEbiography}
\begin{IEEEbiography}
[{\includegraphics[width=1in,height=1.25in,clip,keepaspectratio]{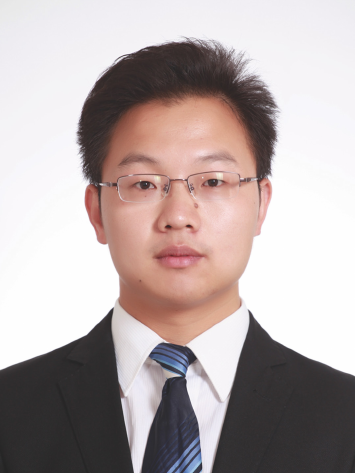}}]
{Keyou You}   received the B.S. degree in Statistical Science from Sun Yat-sen University, Guangzhou, China, in 2007 and the Ph.D. degree in Electrical and Electronic Engineering from Nanyang Technological University (NTU), Singapore, in 2012. After briefly working as a Research Fellow at NTU, he joined Tsinghua University in Beijing, China where he is now an Associate Professor in the Department of Automation. He held visiting positions at Politecnico di Torino, The Hong Kong University of Science and Technology, The University of Melbourne and etc. His current research interests include networked control systems, distributed  algorithms, and their applications.

Dr. You received the Guan Zhaozhi award at the 29th Chinese Control Conference in 2010 and a CSC-IBM China Faculty Award in 2014. He was selected to the National 1000-Youth Talent Program of China in 2014 and received the National Science Fund for Excellent Young Scholars in 2017.\end{IEEEbiography}

\begin{IEEEbiography}
[{\includegraphics[width=1in,height=1.25in,clip,keepaspectratio]{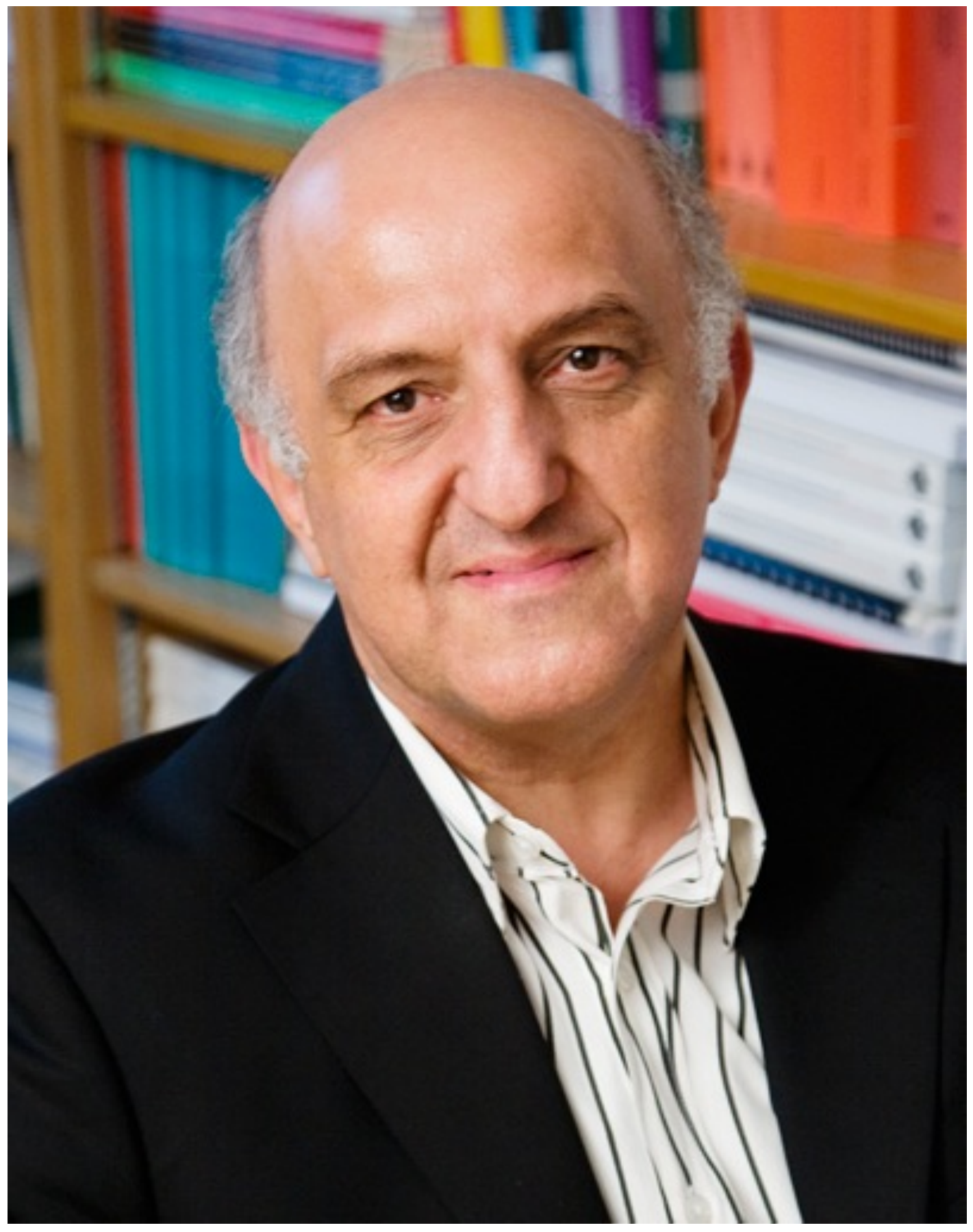}}]
{\bf Tamer Ba\c{s}ar} (S'71-M'73-SM'79-F'83-LF'13) is with the University of Illinois at Urbana-Champaign, where he holds the academic positions of  Swanlund Endowed Chair;    
Center for Advanced Study Professor of  Electrical and Computer Engineering; 
Research Professor at the Coordinated Science
Laboratory; and Research Professor  at the Information Trust Institute. 
He is also the Director of the Center for Advanced Study.
He received B.S.E.E. from Robert College, Istanbul,
and M.S., M.Phil, and Ph.D. from Yale University. 

He is a member of the US National Academy
of Engineering,  member of the  European Academy of Sciences, and Fellow of IEEE, IFAC (International Federation of Automatic Control) and SIAM (Society for Industrial and Applied Mathematics), and has served as president of IEEE CSS (Control Systems  Society), ISDG (International Society of Dynamic Games), and AACC (American Automatic Control Council). He has received several awards and recognitions over the years, including the
highest awards of IEEE CSS, IFAC, AACC, and ISDG, the IEEE Control Systems Award, and a number of international honorary doctorates and professorships. He has over 800 publications in systems, control, communications, networks,
and dynamic games, including books on non-cooperative dynamic game theory, robust control,
network security, wireless and communication networks, and stochastic networked control. He was
the Editor-in-Chief of Automatica between 2004 and 2014, and is currently  editor of several book series. His current research interests include stochastic teams, games, and networks; security; and cyber-physical systems.
\end{IEEEbiography}

\end{document}